\documentclass{article}
\usepackage{a4wide}
\usepackage{amssymb,amsfonts,amsmath,amsthm}
\usepackage{tikz}
\usepackage[bookmarks=true,hyperindex=true]{hyperref}
\hypersetup{
  pdftitle={Bipartite secret sharing and staircases},
  pdfauthor={Laszlo Csirmaz, Frantisek Matus, and Carles Pado},
  pdfsubject={A new characterization of Shannon-ideal multipartite access structures is given
which offers a straightforward and simple approach to describe ideal
bipartite and tripartite access structures. The Shannon complexity is
determined for a range of bipartite access structures, including those
determined by two points, staircases with equal widths and heights, and
staircases with all heights 1. Matching linear schemes are presented
for some non-ideal cases, including staircases where all heights are 1 and
all widths are equal. Finally, finding the Shannon complexity of a bipartite
access structure can be considered as a discrete submodular optimization
problem.},
  pdfkeywords={cryptography; multipartite secret sharing; entropy method;
linear secret sharing; submodular optimization},
  bookmarksopen=false,
  plainpages=false,
}

\DeclareMathAlphabet{\mathsfsl}{OT1}{cmss}{m}{sl}
\def\w{{\mathsfsl w}}
\def\h{{\mathsfsl h}}
\def\H{{\mathbf H\mskip0.5\thinmuskip}}
\def\M{{\mathcal M}}

\def\NN{{\mathbb N}}
\def\F{{\mathbb F}}
\def\FF{{\mathcal F}}
\def\G{{\mathcal G}}
\def\R{{\mathbb R}}
\def\Ing{\mathop{\hbox{\sf Ing}\mskip0.5\thinmuskip}}
\def\dim{\mathop{\hbox{\sf dim}}}
\def\cl{\mathop{\hbox{\sf cl}\mskip0.5\thinmuskip}}
\def\qed{$\square$}

\let\phi\varphi
\newtheorem{theorem}{Theorem}[section]
\newtheorem{corollary}[theorem]{Corollary}

\newtheorem{lemma}[theorem]{Lemma}
\newtheorem{proposition}[theorem]{Proposition}
\theoremstyle{definition}
\newtheorem{definition}[theorem]{Definition}
\newtheorem*{optproblem}{Optimization Problem}
\makeatletter
\def\hstrut{\vrule\@width\z@\@depth\z@\@height\ht\strutbox}
\makeatother
\let\oldfrac=\frac
\def\frac#1#2{\oldfrac{\;\hstrut#1\;}{\;\hstrut#2\;}}

\title{\bf Bipartite secret sharing and staircases\footnote{%
This paper is based on, and is an extension of, two unfinished
manuscripts of our late colleague Fero Mat\'u\v s. The bulk of the 
work reported in this paper was done in Prague where Fero hosted the
two other authors.}}
\date{}
\author{Laszlo Csirmaz\thanks{R\'enyi Institute, Budapest and
UTIA, Prague}
\and \fbox{Franti\v sek Mat\'u\v s\thanks{UTIA, Prague}~} 
\and Carles Padr\'o\thanks{Polytechnic University of Catalonia, Barcelona}}

\begin{document}

\maketitle
\begin{abstract}

\noindent Bipartite secret sharing schemes have a bipartite access structure
in which the set of participants is divided into two parts and all
participants in the same part play an equivalent role. Such a bipartite
scheme can be described by a \emph{staircase}: the collection of its minimal
points. The complexity of a scheme is the maximal share size relative to the
secret size; and the $\kappa$-complexity of an access structure is the best lower
bound provided by the entropy method. An access structure is $\kappa$-ideal if it
has $\kappa$-complexity 1. Motivated by the abundance of open problems in
this area, the main results can be summarized as follows. First, a new
characterization of $\kappa$-ideal multipartite access structures is given
which offers a straightforward and simple approach to describe ideal
bipartite and tripartite access structures. Second, the $\kappa$-complexity is
determined for a range of bipartite access structures, including those
determined by two points, staircases with equal widths and heights, and
staircases with all heights 1. Third, matching linear schemes are presented
for some non-ideal cases, including staircases where all heights are 1 and
all widths are equal. Finally, finding the Shannon complexity of a bipartite
access structure can be considered as a discrete submodular optimization
problem. An interesting and intriguing continuous version is defined which
might give further insight to the large-scale behavior of these optimization
problems.



\end{abstract}


\section{Introduction}
Secret sharing schemes serve as a natural cryptographic primitives used
in group signatures, secure file storage and secure multiparty computation
just to mention a few applications. The initial idea goes back to
Blakley \cite{blakley} and Shamir \cite{shamir}. For an introduction and early bibliography
see \cite{simmons}, for a more recent review see
\cite{beimel-survey} or \cite{blakley-kabaty}.

A \emph{secret sharing scheme}, abbreviated as sss, 
involves a \emph{secret} $0$, the set
$[n]=\{1,2,\dots,n\}$ of $n$ \emph{participants} and an \emph{access structure}
$\Gamma$, which is a nonempty family of subsets of $[n]$ that is closed to supersets
and does not contain the empty set. A participant $i\in [n]$ is
\emph{essential} if some set $I\in\Gamma$ contains $i$ while $I\setminus
i\notin \Gamma$.

In the traditional \emph{probabilistic framework} a secret sharing scheme
consists of jointly distributed random variables $\xi_0,\xi_1, \dots,\xi_n$
taking finitely many values such that $\xi_0$ -- the secret -- is
a function of the vector $\xi_I=\langle\xi_i$, $i\in I\rangle$ almost surely
if and only if $I\in\Gamma$. In an application the dealer samples the
distribution, sets the secret value to be that of $\xi_0$, and communicates
the share $\xi_i$ to participant $i$ privately. The condition ensures that
based on their shares only, participants in $I$ can recover the secret almost
surely if and only if $I\in\Gamma$.

In a \emph{linear framework} the secret sharing scheme consists of subspaces $E_0,E_1,
\dots,E_n$ of some finite linear space $V$ such that $E_0$ is contained in $E_I$
if and only if $I\in \Gamma$, where $E_I$ is the linear span of $\bigcup\{E_i:i\in
I\}$. Any linear sss can be turned into the probabilistic one by setting
$\xi_i$ to be the orthogonal projection of a randomly chosen vector from $V$
to $E_i$.

Widening the usual definition we consider also a \emph{polymatroidal
framework} in which case the scheme consists of a polymatroid $(0\cup[n],h)$
such that $h(I)=h(0\cup I)$ for each $I\in\Gamma$, see, for example,
\cite{padro:optimization}. The scheme $(0\cup[n],h)$ \emph{realizes}
$\Gamma$ if $h(0\cup I)=h(I)$ holds if and only if $I\in\Gamma$.

Polymatroidal schemes cover both probabilistic and linear ones.
In the first case 
the rank $h(I)$ equals the Shannon entropy $\H(\xi_I)$, and 
in the second case $h(I)=\dim(E_I)$. 
Polymatroids obtained this way are called entropic and linear, respectively.

General polymatroidal schemes are not associated with any sort of practical
realization as opposed to probabilistic and linear schemes. Their principal
role is to provide a common platform to investigate the possibilities and
limits of the entropy method using only Shannon information
inequalities \cite{MaPa10}. As polymatroidal schemes cover both
probabilistic and linear ones, results on polymatroidal schemes
(such as lower bounds) automatically carry over to the realizable cases.
From now on, if not mentioned otherwise, all schemes are polymatroidal ones.

The scheme $(0\cup[n],h)$ is \emph{perfect} if $h(0)$ is positive and for
every $I\subseteq [n]$ either $h(0\cup I)=h(0)+h(I)$ or $h(0\cup I)=h(I)$.
Intuitively this means that any collection $I$ of participants either
determines the secret, or has no information on the secret. All schemes in
this paper are assumed to be perfect. Every access structure can be realized
by some perfect linear scheme, see \cite{ito}.

Referring to the probabilistic framework, $h(0)$ is the \emph{secret size}
while $h(i)$ is the \emph{share size} of the participant $i\in[n]$. In a
perfect scheme $h(i)\ge h(0)$ for essential participants $i\in[n]$.

The \emph{information ratio}, or \emph{complexity}, of a scheme $(0\cup[n],h)$ is
the largest share size relative to the secret size. In other words, it is
the maximal value of $h(i)/h(0)$ as $i$ runs over the elements of $[n]$. For
an access structure $\Gamma$ its \emph{Shannon complexity}, denoted by
$\kappa(\Gamma)$, is the infimum of the information ratios of all of its
perfect polymatroidal realizations \cite{MaPa10}. When the minimization is restricted to entropic, or linear
polymatroids, the corresponding infimum is denoted by $\sigma(\Gamma)$, or
$\lambda(\Gamma)$, respectively. It follows that $\lambda(\Gamma) \ge
\sigma(\Gamma) \ge \kappa(\Gamma) \ge 1$. By \cite{Csirmaz1997} we have
$\kappa(\Gamma) \le n$ for all $\Gamma$ and it can happen that
$\kappa(\Gamma) \ge O(n/\log n)$. The infimum is always achieved for
$\kappa(\Gamma)$ (thus it is actually a minimum); while there is an access
structure $\Gamma$ for which $\lambda(\Gamma)=\sigma(\Gamma)=1$ and neither
of these values are taken, see \cite{Matus18}. A sss is \emph{$\kappa$-ideal} if
its Shannon complexity is the smallest possible, namely 1. The access structure
$\Gamma$ is \emph{ideal} if it is realized by some ideal entropic
polymatroid, that is, $\sigma(\Gamma)=1$.
A fundamental result by Brickell and Davenport \cite{brickell-davenport}
essentially states that a $\kappa$-ideal sss is actually a matroid (up
to a scaling factor), which is determined uniquely by the access
structure when all participants are essential. 
Therefore, a $\kappa$-ideal access structure is ideal if and only
if that matroid has an entropic multiple.
Matroid ports, a combinatorial object introduced by Lehman \cite{Leh64},
are basically the same as $\kappa$-ideal access structures.
The characterization of matroid ports by Seymour~\cite{Sey76}
implies that $\kappa(\Gamma) \ge 3/2$ if $\Gamma$ is not 
$\kappa$-ideal~\cite{MaPa10}.

An access structure $\Gamma$ is \emph{threshold} if $I\in\Gamma$ depends
only on the cardinality of $I$. 
Even though only threshold secret sharing schemes
are needed in most applications, 
some situations, as for example hierarchical organizations, 
require access structures in which participants
are divided into several groups according to 
their different roles.
Specifically, $\Gamma$ is
\emph{multipartite}, if $[n]$ can be expressed as a disjoint union of $m\ge
1$ sets $N_1$, \dots, $N_m$ such that $I\in\Gamma$ is
determined solely through the cardinalities of $I\cap N_1$ through $I\cap
N_m$. For $m=2$ such an access structure is called \emph{bipartite}. If $n_1$,
$n_2$ denotes the number of elements in $N_1$ and $N_2$, respectively, a
bipartite $\Gamma$ gives rise to an integer $\ell=\ell_\Gamma\ge 1$ and two
sequences of integers
\begin{equation}\label{eq:1}
 0\le i_1< \cdots < i_\ell\le n_1
  ~~~ \mbox{ and } ~~~
 n_2\ge j_1 > \cdots > j_\ell \ge 0,
\end{equation}
such that $I\in\Gamma$ is equivalent to $|I\cap N_1|\ge i_k$ and $|I\cap
N_2|\ge j_k$ for some $1\le k \le \ell$. The sequence $(i_1,j_1)$, \dots,
$(i_\ell,j_\ell)$ is a \emph{staircase} determining $\Gamma$, having steps
of width $\w_k=i_{k+1}-i_k$, and heights $\h_k=j_k-j_{k+1}$.
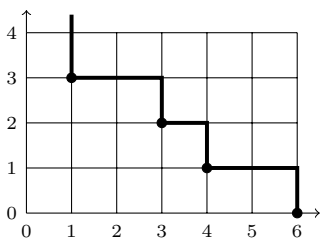
\begin{figure}[htb]
\begin{center}\begin{tikzpicture}[scale=0.6]
\draw(0,0) grid (6,4);
\draw[->] (0,4)--(0,4.5); \draw[->] (6,0)--(6.5,0);
\foreach\x in{0,...,6}{\draw (\x,-0.07) node[below]{$\scriptstyle\x$};}
\foreach\y in{0,1,2,3,4}{\draw(0,\y) node[left]{$\scriptstyle\y$};}
\draw[fill] (1,3) circle(3pt) (3,2) circle(3pt)
 (4,1) circle(3pt)   (6,0) circle(3pt);
\draw[line width=1.5pt](1,4.4)--(1,3)--(3,3)--(3,2)--(4,2)--(4,1)--(6,1)--(6,0);
\end{tikzpicture}\end{center}
\vskip -18pt

\caption{An example staircase}\label{fig:staircase-example}
\end{figure}
The staircase is \emph{regular} if all widths are the same and all heights
are the same. Figure \ref{fig:staircase-example} illustrates a staircase of
length $\ell_\Gamma=4$ for $\Gamma$ bipartite with $n_1=6$ and $n_2=4$. The
widths of the steps are 2, 1, 2, and the heights are 1, 1, 1. We refer the
reader to the works \cite{padro:idealmulti, padro:optimization,
padro:bipartite} for further motivation, basic definitions, and a more
gentle introduction to this topic.

\subsection*{Main results}

Motivated by the open problems raised in \cite{padro:optimization}, the main
results of this paper can be summarized as follows.

a) In Section \ref{sec:multipartite}
a complete characterization of $\kappa$-ideal multipartite access structures
is presented.
It simplifies both the description and the proof of the
characterization in~\cite[Theorem~5.3]{padro:idealmulti}.
 Using the fact that in the bipartite and tripartite cases
ideal and $\kappa$-ideal access structures coincide,
we recover the results in~\cite{padro:idealmulti} about
ideal bipartite and tripartite access structures. While they used
detailed case by case analysis, our result leads to the same collection of 
access structures in a simpler way, explaining the occurence of exceptional
cases. The method is illustrated for the bipartite case; the similar procedure
for the tripartite case is left to the interested reader.
In addition, by using Ingleton inequality,
we present a much simpler proof for the fact that
ideal and $\kappa$-ideal tripartite access structures coincide.
We hope that the general characterization
gives further insight into these interesting families.

b) In Section \ref{sec:shannon-compl} the Shannon complexity
$\kappa(\Gamma)$ is determined
for a range of bipartite access structures including those determined by two
points (that is, when $\ell_\Gamma=2$), regular staircases with equal width and
height, and staircases where all heights are $\h_k=1$ and the width sequence
satisfies an additional technical assumption. 
Shannon complexity is the best lower bound on $\sigma(\Gamma)$ implied by
the Shannon information inequalities. Computing $\kappa(\Gamma)$ is a linear
optimization problem with linear constraints. The number of constraints, 
however, is exponential in the number of participants. The exact value of
the Shannon complexity was known only for a few infinite families of graph-based 
access structures \cite{bssv,Cs-Tardos,Cs-cubes}. Exploiting the internal
symmetry of bipartite access structures their Shannon complexity is
expressed as the solution of another linear optimization problem where the
number of constraints is at most quadratic in the number of participants.
The solution of the reduced optimization problem is determined for the above
bipartite access structures using the duality theorem of linear programming.
While it is known that the Shannon complexity cannot exceed the number $n$ of
participants \cite{Csirmaz1997}, it is an open problem whether there is a
positive constant $c>0$ such that $\kappa(\Gamma)>c\cdot n$ for some
bipartite access structure $\Gamma$ on $[n]$ for infinitely many $n$.

c) In a few cases linear schemes were found matching the corresponding
Shannon bound, such as all regular staircases of height $1$. The
constructions are presented in Section \ref{sec:linear-examples}.

d) Determining the Shannon complexity of bipartite access structures can be
considered as discrete submodular optimization, see \cite{fujishige,hazan-kale}.
In Section \ref{sec:continuous} a corresponding continuous optimization
problem is defined. The intuition is scaling down the increasing
optimization problems so that constraints separating qualified and
unqualified subsets converge to constraints along a continuous curve. 
Results on the this continuous optimization problem could give a hint on the
large scale behavior of bipartite access structures.
For the continuous case, a local lower bound on the optimal value is proved
which is tight in certain cases. This section concludes with some open
problems in this framework.


\section{Ideal multipartite secret sharing revisited}\label{sec:multipartite}

Ideal multipartite secret sharing schemes received a considerable attention,
see \cite{padro:idealmulti} and the references therein. This section
provides a characterization of $\kappa$-ideal multipartite access
structures. It is indicated how this characterization can be used to
generate all ideal bipartite and tripartite access structures. The main tool
is the fundamental result by Brickell and Davenport \cite{brickell-davenport}, 
namely, a $\kappa$-ideal sss is a matroid (up
to a scaling factor), and this matroid is determined uniquely by the access
structure. Moreover a $\kappa$-ideal access structure is ideal if and only
if this matroid has an entropic multiple.

\smallskip

Let $(M,f)$ be an integer polymatroid which is linearly representable by
subspaces of a finite dimensional vector space over the finite field $\F$. As the same
representation works over any extension of $\F$, we may assume $\F$ to be
arbitrarily large. Fix $k\in M$, and let $E_k$ be the subspace of dimension
$f(k)$ corresponding to $k$ witnessing the linear representability. Choose a
``generic'' vector $u\in E_k$ which is not in any proper subspace of $E_k$
cut off by the subspaces $E_J$ for $J\subseteq M$. Each such requirement
discards at most $|\F|^{f(k)-1}$ elements of $E_k$, thus such a generic
vector exists whenever $\F$ is large enough. It is clear that for any
$J\subseteq M$ the linear span of $u\cup E_J$ has dimension one more than
the dimension of $E_J$ except when $E_k$ is a subspace of $E_J$. This
motivates the following definition. Fix $k\in M$, $u\notin M$, and extend
the rank function $f$ to subsets of $u\cup M$ as follows: for every
$J\subseteq M$ let
$$
    f(uJ) = \begin{cases}
           f(J)   & \mbox{ if } f(kJ)=f(J),\\
           f(J)+1 & \mbox{ otherwise}.
    \end{cases}
$$
The clearly integer polymatroid $(u\cup M,f)$ is called the \emph{generic extension
of $f$ along $k\in M$}. The above discussion shows that whenever $(M,f)$ is
linearly representable then so is this generic extension.

\smallskip

Let $(0\cup[n],h)$ be a $\kappa$-ideal sss realizing the multipartite access
structure $\Gamma$ with partition $[n]=N_1\cup\dots\cup N_m$. By the
result of Brickell and Davenport \cite{brickell-davenport} $h$ can be
assumed to be a matroid which is invariant under every permutation $\pi$ of $[n]$ that keeps $N_k$ fixed for each $k \in [m]$. Let $([n],h')$ be the
restriction of $h$ discarding the secret 0. As $h$ is a one-point
extension of $h'$, it is determined uniquely by the \emph{modular cut} (see,
e.g., \cite{oxley})
\begin{equation}\label{eq:N}
   \FF=\{ F\subseteq [n]: \mbox{ $F$ is a flat in $([n],h')$ and }h(F)=h(0F) \}.
\end{equation}
\begin{lemma}\label{lemma:fero1}
If $F$ is a minimal flat in $\FF$ and $k\in[m]$, 
then either $N_k\subseteq F$
or $F \cap N_k = \emptyset$.
\end{lemma}
\begin{proof}
Suppose that $F\cap N_k$ is neither empty nor equals $N_k$. 
Let $\pi$ be the permutation of $[n]$ which swaps only two
elements of $N_k$, one in $F\cap N_k$ and the other in $N_k\setminus F$. As
the matroid $h$ is invariant under $\pi$, both $F$ and $\pi(F)$ 
are flats with the same rank and $\pi(F)\in\FF$. Observe that
$F$ and $\pi(F)$ form a modular pair. This is so as
$F\cup\pi(F)$ has one element more than $F$, and its rank is strictly bigger
than that of $F$ (as $F$ is a flat), thus equals $h(F)+1$. Similarly,
$F'=F\cap \pi(F)$ has one element less than $F$, and its rank must be
strictly smaller than the rank of $F$ (as $F'$ is an intersection of two
different flats), thus $h(F')=h(F)-1$. As both $F$ and $\pi(F)$ are in $\FF$,
their intersection, $F'$ is in $\FF$ as well. 
That contradicts the assumption that $F$ is minimal in $\FF$. 
\end{proof}

\begin{lemma}\label{lemma:fero2}
Let $i\in N_k$ and $i\notin J\subseteq 0\cup[n]$. If $h(J)\not=h(J\cup N_k)$,
then $h(iJ)=h(J)+1$.
\end{lemma}
\begin{proof}
Suppose $h(iJ)=h(J)$. By the multipartite symmetry the same equality
holds for every $i\in N_k\setminus J$, and then $h(J)=h(J\cup N_k)$.
\end{proof}

For $I\subseteq 0\cup[m]$ the set $\bigcup\{N_i:i\in I\}$ is denoted by
$N_I$ where we take $N_0=\{0\}$. 

\begin{lemma}\label{lemma:indep-set}
$J\subseteq 0\cup[n]$ is independent in $(0\cup[n],h)$ if and only if $|J\cap N_I|\le
h(N_I)$ for all $I\subseteq 0\cup[m]$.
\end{lemma}
\begin{proof}
The condition is clearly necessary. Sufficiency is immediate for $m=1$
as in this case $(0\cup[n],h)$ is the uniform matroid.
Otherwise let $B=J\cap N_m$. If $B=\emptyset$, then use induction on the 
matroid restricted to $0\cup[n]\setminus N_m$. If $B\neq\emptyset$, then
$|B| \le h(N_m)$ by assumption, thus $h(B)=|B|$ by the multipartite
symmetry. From here induction on the contraction 
$(0\cup[n],h)\setminus N_m$
gives the claim of the lemma.
\end{proof}
Since the collection of independent sets determines the matroid \cite{oxley},
a consequence of this lemma is that 
the matroid $(0\cup[n],h)$ is uniquely determined by the ranks
$\{h(N_I): I\subseteq 0\cup[m]\}$.

\begin{lemma}\label{lemma:1-1}
For a partition 
$[n] = N_1 \cup \cdots \cup N_m$, 
there is a one-to-one correspondence between 
the $\kappa$-ideal $m$-partite sss $(0\cup[n],h)$ and the pairs 
$\langle ([m],f'), \M \rangle$, 
where $([m],f')$ is an integer polymatroid 
with $f(k) \le |N_k|$ for each $k \in [m]$
and $\M$ is a modular cut in $([m],f')$.
\end{lemma}
\begin{proof}
Consider the map $\phi$ from $0\cup[n]$ to $0\cup[m]$
defined by $\phi(0)=0$ and $\phi(i)=k$ whenever $i\in N_k$.
Let $(0\cup[n],h)$ be a $\kappa$-ideal $m$-partite sss
for the given partition.
The \emph{factor} of $(0\cup[n],h)$ by $\phi$ is the (integer) polymatroid on the ground set
$0\cup [m]$ with the rank function $f(I)= h(\phi^{-1}(I))$. In
particular, $f(0)=1$ and $f(k)=h(N_k)\le |N_k|$ for $k\in[m]$.
Let $([m],f')$ be the restriction of this polymatroid to $[m]$, and
$\M=\phi(\FF)$ where $\FF$ is the modular cut in (\ref{eq:N}). 
By Lemma \ref{lemma:fero1} $\M$ is a modular cut in
$([m],f')$; this defines the corresponding 
integer polymatroid and modular cut.
In the other direction, take the integer polymatroid $([m],f')$ and the
modular cut $\M$. Let the corresponding one-element extension be
$(0\cup[m],f)$, namely
\begin{equation}\label{eq:ext}
    f(0J) = \begin{cases} f'(J) & \mbox{if $\cl(J)\in \M$,}\\
                               f'(J)+1 &\mbox{otherwise,}
                 \end{cases}
\end{equation}
where $\cl(J)$ is the closure of $J$ in $([m],f')$. 
For the chosen partition of $[n]$, 
take any $m$-partite secret sharing matroid
$(0\cup[n],h)$ such that its $\phi$-factor is $(0\cup[m],f)$.
According to Lemmas \ref{lemma:fero2} and \ref{lemma:indep-set} the ranks of
$(0\cup[n],h)$ are determined uniquely, thus there is at most one such
matroid. To show the existence, starting from $(0\cup[m],f)$ 
take $|N_i|$ generic extensions repeatedly along $i$ for each $i\in[m]$, and
then restrict the final extension to $0\cup[n]$. It is easy to check that 
it has the desired properties.
\end{proof}

The fact that every integer polymatroid is a factor of a matroid goes back
to T.~Helgason \cite{helgason}. A similar construction using a completely
different setting appeared in \cite{bergman-fan}.

The correspondence expressed in Lemma \ref{lemma:1-1} can be turned into a procedure
which enumerates all $\kappa$-ideal access structures. The correctness of
the procedure is immediate from the lemma.

\begin{theorem}\label{thm:generating-kappa-ideal}
The procedure outlined below generates all $\kappa$-ideal $m$-partite
access structures on $[n]=N_1 \cup \cdots \cup N_m$.
\end{theorem}
\begin{enumerate}
\item Take any integer polymatroid $([m],f')$ with $f'(k)\le
|N_k|$, and take a modular cut $\M$ in $([m],f')$.
\item Let $(0\cup[m],f)$ be the
corresponding one-point extension as defined in (\ref{eq:ext}).
\item Starting from $(0\cup[m],f)$ add $N_i$ generic elements along
$i$ for each $i\in [m]$. Restrict the final polymatroid to $0\cup[n]$. 
The result is a matroid $(0\cup[n],h)$; it is $\kappa$-ideal, $m$-partite, and
the corresponding access strcuture is $\{ I\subseteq [n]: h(I)=h(0I) \,\}$.
\end{enumerate}

Note that if the polymatroid $(0\cup[m],f)$ is linearly representable, then
the same applies to the matroid ($0\cup[n],h)$. Consequently the corresponding
access structure can be realized by an ideal linear sss.

\subsection{Ideal bipartite access structures}\label{subsec:ideal-bipartite}

For bipartite access structures the procedure of Theorem
\ref{thm:generating-kappa-ideal} can be detailed as follows.
Take an integer polymatroid $(M,f')$ on the two-element set $M=\{1,2\}$. The
polymatroid
$(M,f')$ is determined by the integer ranks $a=f'(1)$, $b=f'(2)$ and $c=f'(12)$,
where $c\le a+b$ (here $12$ is the two-element set $\{1,2\}$). Assume
neither $a$ nor $b$ is zero and $a,b<c$ (thus $1$, $2$ and $12$ are
all flats).
If $c<a+b$, then $(M,f')$ has four non-trivial modular cuts:
\begin{align*}
  \M_1 &= \{ 1, 12 \},  & \M_2 &= \{ 2, 12 \}, \\
  \M_3 &= \{ 12 \},     & \M_4 &= \{1, 2, 12 \},
\end{align*}
If $c=a+b$, then $\{1,2\}$ is a modular pair, thus $\M_4$ is not a
modular cut. The one-point extensions $(\{0,1,2\},f)$ are
integer polymatroids on three elements, consequently they are linearly
representable; see \cite{matus-studeny}. The generic extensions created in step 3
are also linearly representable, thus every $\kappa$-ideal bipartite access
structure admits an ideal linear sss.

Let us compute the ranks in the generic extension $(0\cup N_1\cup
N_2,h)$. For $I_1\subseteq N_1$ and $I_2\subseteq N_2$ we have
\begin{align*}
  h(I_1) &= \min\{ |I_1|,a \}, \\
  h(I_2) &= \min\{ |I_2|,b \}, \\
  h(I_1\cup I_2) &=\min\{h(I_1)+h(I_2),c\}.
\end{align*}
If $f$ was generated by $\M_1$, then
\begin{alignat*}{2}
  h(0\cup I_1) &=h(I_1)  &&~\Longleftrightarrow ~ h(I_1)=a, \\
  h(0\cup I_2) &=h(I_2)+1 &&~\\
  h(0\cup I_1\cup I_2) &=h(I_1\cup I_2) &&~\Longleftrightarrow ~ h(I_1)=a
  \mbox{ or } h(I_1\cup
  I_2)=c,
\end{alignat*}
and similarly for the other cases. In summary, the access structures 
corresponding to  the modular cuts $\M_1,\M_2,\M_3,\M_4$ are:
\begin{align*}
 \Gamma_1& =\{I_1\cup I_2:~  |I_1|\ge a \mbox{ or }
     ( |I_1|\ge c-b \mbox{ and } |I_1|+|I_2|\ge c ) \},\\
 \Gamma_2&=\{I_1\cup I_2:~ |I_2|\ge b \mbox{ or }
     ( |I_2|\ge c-a \mbox{ and } |I_1|+|I_2|\ge c ) \},\\
 \Gamma_3&=\{I_1\cup I_2:~ |I_1|\ge c-b \mbox{ and }
      |I_2|\ge c-a \mbox{ and }  |I_1|+|I_2|\ge c \}, \\
 \Gamma_4&=\{I_1\cup I_2: ~
      |I_1|\ge a \mbox{ or } |I_2|\ge b \mbox{ or } |I_1|+|I_2|\ge c\}.
\end{align*}
If $c=a+b$ then $\Gamma_4$ is missing as it would be the same access
structure which is generated by $a$, $b$, and $c-1$. Figure \ref{fig:ideal2}
illustrates the four types of ideal bipartite access structures. Qualified
subsets correspond to the lattice points in the shaded area.

\begin{figure}[htp]
\begin{center}\begin{tikzpicture}[scale=0.3]
\begin{scope}[xshift=0cm]
\fill[gray!20] (1,5.4)--(1,4)--(3,2)--(3,0)--(5.4,0) to[out=90,in=0]
(1,5.4);
\draw[dotted,very thin] (0,4)--(1,4);
\draw[->] (0,0)--(0,5.5);
\draw[->] (0,0)--(5.5,0);
\draw[thin] (0,5)--(5,0);
\draw (0,4) node[left] {$b$} (3,0) node[below] {$a$} (2,2.2) node {$c$};
\draw[very thick] (1,5.4)--(1,4)--(3,2)--(3,0)--(5.4,0);
\draw (0.2,0.3) node [below left] {$\Gamma_1$};
\end{scope}
\begin{scope}[xshift=9cm]
\fill[gray!20] (0,5.4)--(0,4)--(1,4)--(3,2)--(5.4,2) to[out=90,in=0]
(0,5.4);
\draw[dotted,very thin] (3,0)--(3,2);
\draw[->] (0,0)--(0,5.5);
\draw[->] (0,0)--(5.5,0);
\draw[thin] (0,5)--(5,0);
\draw (0,4) node[left] {$b$} (3,0) node[below] {$a$} (2,2.2) node {$c$};
\draw[very thick] (0,5.4)--(0,4)--(1,4)--(3,2)--(5.4,2);
\draw (0.2,0.3) node [below left] {$\Gamma_2$};
\end{scope}
\begin{scope}[xshift=18cm]
\fill[gray!20](1,5.4)--(1,4)--(3,2)--(5.4,2) to[out=90,in=0] (1,5.4);
\draw[dotted,very thin] (0,4)--(1,4);
\draw[dotted,very thin] (3,0)--(3,2);
\draw[->] (0,0)--(0,5.5);
\draw[->] (0,0)--(5.5,0);
\draw[thin] (0,5)--(5,0);
\draw (0,4) node[left] {$b$} (3,0) node[below] {$a$} (2,2.2) node {$c$};
\draw[very thick] (1,5.4)--(1,4)--(3,2)--(5.4,2);
\draw (0.2,0.3) node [below left] {$\Gamma_3$};
\end{scope}
\begin{scope}[xshift=27cm]
\fill[gray!20] (0,5.4)--(0,4)--(1,4)--(3,2)--(3,0)--(5.4,0) to[out=90,in=0]
(0,5.4);
\draw[->] (0,0)--(0,5.5);
\draw[->] (0,0)--(5.5,0);
\draw[thin] (0,5)--(5,0);
\draw (0,4) node[left] {$b$} (3,0) node[below] {$a$} (2,2.2) node {$c$};
\draw[very thick] (0,5.4)--(0,4)--(1,4)--(3,2)--(3,0)--(5.4,0);
\draw (0.2,0.3) node [below left] {$\Gamma_4$};
\end{scope}
\end{tikzpicture}\end{center}

\kern-14pt
\caption{Ideal bipartite access structures corresponding to the modular cuts}\label{fig:ideal2}
\end{figure}

\subsection{Ideal tripartite access structures}

Tripartite $\kappa$-ideal access structures can be generated similarly to
the bipartite case. One starts from an integer polymatroid $(M,f')$ on three
elements, extends it to $(0\cup M,f)$ using a modular cut, and then adds
generic elements. As $(M,f')$ is on three elements, it is linearly
representable. We claim that $(0\cup M,f)$ is also linearly representable,
thus all tripartite $\kappa$-ideal access structures are, in fact, ideal.
This claim has been proved first in \cite[Theorem 19]{padro:idealmulti}.

An integer polymatroid on four elements $\{a,b,c,d\}$ has a linearly 
representable multiple if and only if it satisfies all instances of the Ingleton
inequality $\Ing(a,b,c,d)\ge 0$, see \cite{matus-studeny}. The Ingleton expression
is a linear combination of ten ranks as follows \cite{ingleton}:
\begin{align}\label{eq:ingleton}
  \Ing(a,b,c,d)&=-f(a)-f(b)-f(cd)-f(abc)-f(abd)+{}\\
               & ~~~~{}+f(ab)+f(ac)+f(ad)+f(bc)+f(bd),\nonumber
\end{align}
where, as usual, brackets around singletons and the union signs are omitted. 
The Ingleton expression is invariant for swapping the first pair and
the second pair of arguments, respectively, which means that it has six different
instances.
The following inequalities hold in every polymatroid:
\begin{align*}
\Ing(a,b,c,d) &+ f(a)+f(c)-f(ac)\ge 0, \\
\Ing(a,b,c,d) &\ge f(a) - f(ac), \\
\Ing(a,b,c,d) &\ge f(c)-f(ac).
\end{align*}
For example, the first inequality can be written equivalently as
$$
    \delta(ab,bc) + \delta(ad,bd)+\delta(c,d)\ge 0,
$$
where $\delta(I,J)=f(I)+f(J)-f(I\cup J)-f(I\cap J)$ is the non-negative
modular defect of $I$ and $J$. Similar rearrangements work for the other two
inequalities.

Let us return to the claim that $(0\cup M,f)$ is linearly representable. Due
to the symmetry of the Ingleton expression we can assume that the secret $0$
is either $a$ or $c$. As $f$ is
integer, $\Ing(a,b,c,d)<0$ means $\Ing(a,b,c,d)\le -1$, and then
$f(a)+f(c)-f(ac)\ge 1$ by the first inequality. Also, $f(0)=1$ implies that
either $f(a)=1$ or $f(c)=1$, and then either $f(a)-f(ac)\ge 0$ or
$f(c)-f(ac)\ge 0$. In both cases $\Ing(a,b,c,d)\ge 0$ according to the
second and third inequality. Consequently all Ingleton expressions are
non-negative proving that $(0\cup M,f)$ is linearly representable, as
claimed.


\section{Definitions and basic tools}\label{sec:definitions}

This section introduces the basic tools which will be used in Section
\ref{sec:shannon-compl} to provide lower bounds on the Shannon complexity of 
some bipartite access structures.

Consider the rank function $f$ of a sss polymatroid
representing a bipartite access structure on $N_1\cup N_2$. All polymatroidal
constraints on the rank function are linear, thus one can incorporate 
all symmetries of the access structure into
the constraints (by taking the average over all automorphisms of the access
structure). In this way the rank function $f(J)$ depends only on the numbers
$|J\cap N_1|$ and $|J\cap N_2|$. This idea is detailed in
\cite{padro:optimization} where it is shown that all machinery can be
explained in terms of so-called {\em multipartite polymatroids}. 
In the bipartite case the rank function $f(i,j)$ is defined on $\NN\times\NN$,
the set of non-negative lattice points. 
Constraints resulting from the polymatroidal axioms are listed in 
(\ref{eq:shannon}) where $i$ and $j$ run over the non-negative integers.
These constraints can also be considered as definition: if a real function
$f$ defined on the non-negative lattice points satisfies all these constraints,
then it is a \emph{discrete submodular function}.
\begin{equation}\label{eq:shannon}
\begin{array}{l@{~~~~~}l}
f(i,j)\ge 0,~~ f(0,0)=0 
 & \mbox{non-negativity} \\[5pt]
\left.\begin{array}{@{}l}
f(i+1,j) \ge f(i,j) \\
f(i,j+1)\ge f(i,j)
\end{array}\right\}
 & \mbox{monotonicity} \\[14pt]
\left.\begin{array}{@{}l}
f(i,j)-f(i-1,j) \ge f(i+1,j)-f(i,j)\\
f(i,j)-f(i,j-1) \ge f(i,j+1)-f(i,j)
\end{array}\right\}
 & \mbox{submodularity - 1} \\[14pt]
f(i+1,j)-f(i,j) \ge f(i+1,j+1)-f(i,j+1)
 & \mbox{submodularity - 2}
\end{array}
\end{equation}
Next to these constraints additional
\emph{strong inequalities} express the additional requirement that the polymatroid 
should be a sss for the access structure $\Gamma$.
It turns out that this requirement is equivalent to require that the
difference between the left and right hand side in some inequalities in
(\ref{eq:shannon}) is at least one, depending on whether (any, or all) of the
subsets $J_{ij}$
identified by the arguments $i$ and $j$, that is, $|J_{ij}\cap N_1|=i$ and
$|J_{ij}\cap
N_2|=j$, is qualified or not. In (\ref{eq:shannon-strong}) below rather than using such a
verbal description, we use the notation $f^\bullet(i,j)$ to indicate
that the subset $J_{ij}$ corresponding to the argument $($i,$j)$ is qualified, and
$f^\circ(i,j)$ to indicate that $J_{ij}$ is unqualified.
\begin{equation}\label{eq:shannon-strong}
\begin{array}{l@{~~~~~}l}
\left.\begin{array}{@{}l}
f^\bullet(i+1,j) \ge f^\circ(i,j) + 1\\
f^\bullet(i,j+1)\ge f^\circ(i,j) + 1
\end{array}\right\}
 & \mbox{strong monotonicity} \\[14pt]
\left.\begin{array}{@{}l}
f^\bullet(i,j)-f^\circ(i-1,j) \ge f^\bullet(i+1,j)-f^\bullet(i,j) +1\\
f^\bullet(i,j)-f^\circ(i,j-1) \ge f^\bullet(i,j+1)-f^\bullet(i,j) +1
\end{array}\right\}
 & \mbox{strong submodularity - 1} \\[14pt]
f^\bullet(i+1,j)-f^\circ(i,j) \ge f^\bullet(i+1,j+1)-f^\bullet(i,j+1) +1
 & \mbox{strong submodularity - 2}
\end{array}
\end{equation}
Let $H$ and $V$ denote the first horizontal and vertical values at the
origin, respectively:
\begin{equation}\label{eq:HV}
 H = f(1,0), ~~~~~ V = f(0,1).
\end{equation}
With these notation the Shannon complexity of the bipartite access structure
$\Gamma$ is
\begin{equation}\label{eq:defkappa}
\kappa(\Gamma) = \inf\nolimits_f \, \{ \,\max(H,V)\,:\,
 \mbox{ $f$ satisfies (\ref{eq:shannon}) and (\ref{eq:shannon-strong}) } \}
.
\end{equation}
The aim of this Section and Section \ref{sec:shannon-compl} is to find, or give a good estimate for, this value.

\smallskip
Let us fix the bipartite access structure $\Gamma$ and a function $f$
which satisfies the constraints in (\ref{eq:shannon}) and
(\ref{eq:shannon-strong}). Arguments of $f$ are the lattice
points in the non-negative quadrant. These points are denoted by $A_1$, $B_2$,
etc., and with an abuse of notation, they also denote the value of
$f$ at that point. Qualified and unqualified arguments are denoted by solid
and hollow dots, respectively. Figure \ref{fig:example} illustrates 
three horizontally consecutive lattice points $A_1$, $A_2$, and $A_3$ such
that $A_1$ is unqualified, and $A_2$ and $A_3$ are qualified. 
Monotonicity constraints from (\ref{eq:shannon}) give
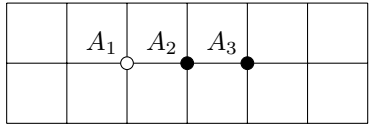
\begin{figure}[htb]
\begin{center}\begin{tikzpicture}[scale=0.8]
\draw(1,2) grid (7,4);
\draw (3,3) node[above left]{$A_1$} (4,3) node[above left]{$A_2$} (5,3)
node[above left]{$A_3$};
\draw[fill=white] (3,3) circle (3pt);
\foreach \x in {4,5} {\draw[fill] (\x,3) circle (3pt); }
\end{tikzpicture}\end{center}

\kern-14pt
\caption{Lattice points representing $f$}\label{fig:example}
\end{figure}%
$$
   A_1 \le A_2 \le A_3.
$$
while the first line of submodularity-1 in (\ref{eq:shannon}) translates to
$$
   2A_2  \ge A_1 + A_3.
$$
As $A_1$ is unqualified, and both $A_2$ and $A_3$ are qualified,
the stronger inequalities from (\ref{eq:shannon-strong}) also hold:
$$
	A_2 \ge A_1+1, ~~\mbox{and} ~~ 2A_2 \ge A_1+A_3+1 
$$
Submodularity-1 actually says that the function $f$, going from
left to right (first line), or going from bottom up (second line), is 
concave.
\begin{figure}[htb]
\begin{center}\begin{tikzpicture}[scale=0.8]
\draw(0,-0.1) node[below]{$A$} (2,-0.1) node[below]{$B$}
   (5,-0.1) node[below]{$C$} (7,-0.1) node[below]{$D$};
\draw(-0.4,0)--(0.7,0) (1.3,0)--(3.2,0) (3.8,0)--(5.7,0) (6.3,0)--(7.4,0);
\draw[dotted] (0.7,0)--(1.3,0) (3.2,0)--(3.8,0) (5.7,0)--(6.3,0);
\draw[fill=white](0,0)circle(3pt) (2,0) circle(3pt) (5,0) circle(3pt)
   (7,0)circle(3pt);
\draw(1,0.1) node[above]{$\overbrace{\hbox to 1.6cm{}}^{\textstyle k}$};
\draw(3.5,0.1) node[above]{$\overbrace{\hbox to 2.05cm{}}^{\textstyle \ell}$};
\draw(6,0.1) node[above]{$\overbrace{\hbox to 1.6cm{}}^{\textstyle m}$};
\end{tikzpicture}\end{center}
\vskip -18pt
\caption{Consequences of concavity}\label{fig:concave}
\end{figure}
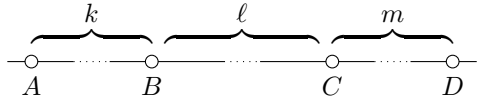
Lemma \ref{lemma:concave} is an easy consequence of this concavity and it refers
to Figure \ref{fig:concave}. The lattice points $A$, $B$, $C$ and $D$ are
on a horizontal (or vertical) line going from left to right (or from bottom up).
The distance between $A$ and $B$, $B$ and $C$,
$C$ and $D$ are $k$, $\ell$, $m$, respectively. In particular,
$k=\ell=m=1$ if $A$, $B$, $C$, $D$ are consecutive nodes.

\begin{lemma}\label{lemma:concave}
With the notation of Figure \ref{fig:concave},
\begin{itemize}
\item[\upshape a)] $\displaystyle
  \frac{B-A}k \ge \frac{C-B}\ell $, 
\item[\upshape b)] $\displaystyle
    \frac{B-A}{k} \ge \frac{D-C}{m}$,
\item[\upshape c)] if $A$ is unqualified, $B$ is qualified, and there are $s$
qualified nodes between $A$ and $B$ ({\em not} including $B$), then
$ \displaystyle \frac{B-A}k \ge \frac{C-B}\ell + \frac{k-s}k$;
\item[\upshape d)] if $A$ and $B$ are unqualified, $C$ and $D$ are qualified, then 
$\displaystyle
    \frac{B-A}{k} \ge \frac{D-C}{m} + 1
$. \hfill\qed
\end{itemize}
\end{lemma}
\noindent
Claim b) is immediate from a); c) is a strong version of a); and d) is a
strong version of b). There are other strong versions of a) and b) depending on
how many qualified nodes are between certain pairs. As these versions can 
also be proved similarly, they will be used without any reference.

\begin{figure}[htb]
\begin{center}\begin{tikzpicture}[scale=0.8]
 \draw[very thin](3,-1)--(4,-1)--(4,-2)--(4.8,-2);
 \draw[very thin](6.2,-2)--(8,-2) (7,-1)--(7,-2);
 \draw[dotted] (4.8,-2)--(6.2,-2);
\draw[line width=1.5pt] 
   (4,-0.5)--(4,-1)--(4.8,-1) (6.2,-1)--(8,-1)--(8,-2);
\draw[line width=1.5pt,dotted](4.8,-1)--(6.2,-1);
\draw[fill=white] (3,-1) circle(3pt);
\foreach \x in {4,7,8}{\draw[fill](\x,-1) circle(3pt); }
\foreach \x in {4,7}{\draw[fill=white](\x,-2) circle(3pt); }
\draw[fill](8,-2)circle(3pt);
\draw(5.99,-0.9) node[above]{$\overbrace{\hbox to 2.93cm{}}^{\textstyle\ell}$};
\draw (2.95,-1) node[below]{$B_1$}
   (4.1,-1) node[below left]{$C_1$} (7.1,-1) node[below left]{$D_1$}
   (4.2,-2) node[below left]{$A_2$} (7.2,-2) node[below left]{$B_2$}
   (8.2,-2) node[below left]{$C_2$};
\end{tikzpicture}\end{center}

\vskip -18pt
\caption{Single increment}\label{fig:onestep}
\end{figure}
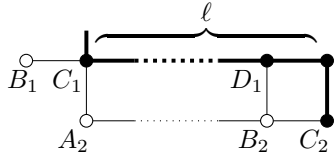
Lemma \ref{lemma:increment} refers to Figure \ref{fig:onestep}. Nodes $C_1$,
$D_1$, and $C_2$ are qualified, and nodes $B_1$, $A_2$, and $B_2$ are not.
The distance between $C_1$ and $D_1$ (between $A_2$ and $B_2$) is $\ell-1$;
$B_1C_1$, $C_1A_2$, etc., have length 1.
\begin{lemma}\label{lemma:increment}
With the notation of Figure \ref{fig:onestep} and assuming $\ell\ge 2$,
\begin{itemize}
\item[\upshape a)]$\displaystyle
  C_1-B_1 \ge (C_2-B_2)+1-\frac{V-1}{\ell-1}
$,
\item[\upshape b)] $\displaystyle
  H\ge C_1-B_1$, and $C_2-B_2\ge 1$.
\end{itemize}
\end{lemma}
\begin{proof}
By claim c) of Lemma \ref{lemma:concave} we have
$$
    C_1-B_1\ge \frac{D_1-C_1}{\ell-1}+1,
$$
and by a) of the same Lemma,
$$
   \frac{B_2-A_2}{\ell-1}\ge C_2-B_2.
$$
To finish the proof one has to observe that $D_1\ge B_2+1$ by strong
monotonicity, and $C_1\le A_2+V$ by submodularity.
\end{proof}

\begin{corollary}\label{corr:simple}
Suppose $\Gamma$ has a step of width $\w=\w_k=i_{k+1}-i_k\ge 2$ such that
$i_k\not= 0$. Then $\kappa(\Gamma)\ge 2-1/\w$.
\end{corollary}
\begin{proof}
\newcommand\ige[1]{\stackrel{\mathrm #1}\ge}
Denote the point $(i_k,j_k)$ by $C_1$, and the point $(i_{k+1},j_k+1)$ by
$C_2$. With this choice Lemma \ref{lemma:increment} gives
$$
  H\ige{(b)} C_1-B_1\ige{(a)} (C_2-B_2)+1-\frac{V-1}{\w-1}
    \ige{(b)} 2-\frac{V-1}{\w-1}.
$$
By rearranging $(\w-1)H+V\ge 2\w-1$, thus either $H$  or $V$ must be at
least $(2\w-1)/\w$, as was claimed.
\end{proof}


\section{Shannon complexity of some bipartite access structures}\label{sec:shannon-compl}

The bound given by Corollary \ref{corr:simple} is tight for some bipartite
access structures,
namely their Shannon complexity is $\kappa=2-1/\w$. To show that
this is the case, it is enough to present a particular submodular function
$f$ on the non-negative grid with $\max\{H,V\}\le\kappa$ which
satisfies all constraints in (\ref{eq:shannon}) and
(\ref{eq:shannon-strong}). Rather than giving the values of $f$ at the grid
points, it is more convenient to give values at horizontal and vertical
edges, which are the difference of the function values at the edge endpoints.
As $f(0,0)$ is zero, these differences determine $f$ uniquely. Properties
(\ref{eq:shannon}) and (\ref{eq:shannon-strong}) can be expressed in terms
of these differences in an equivalent form:
\begin{equation}\label{eq:shannon-diffs}
  \mbox{\begin{tabular}{p{0.6\textwidth}@{~~~~~~}l}
\hangindent 16pt (a) edge values are non-negative,\rule[-6pt]{0pt}{0pt} 
  & monotonicity \\ [4pt]
\hangindent 16pt (b) on each $1\times 1$ square, the sum of left and top 
edges equals the sum of bottom and right edges,\rule[-6pt]{0pt}{0pt} 
  & consistency \\
\hangindent 16pt (c) values are decreasing from left to right, and from 
bottom up (both for vertical and horizontal edges),\rule[-6pt]{0pt}{0pt}
   & submodularity \\
\hangindent 16pt (d) an edge between a qualified and an unqualified vertex
has value at least 1,\rule[-6pt]{0pt}{0pt} 
   & strong monotonicity \\
\hangindent 16pt (e) the increment between two adjacent horizontal
(vertical) edges is at least 
one if the second edge has two qualified endpoints, and the first edge has
only one,\rule[-6pt]{0pt}{0pt}
   & strong submodularity - 1 \\
\hangindent 16pt (f) in an $1\times1$ square with three qualified nodes the
left edge is at least 1 more than the right edge.\rule[-6pt]{0pt}{0pt}
  & strong submodularity - 2
 \end{tabular}}
\end{equation}
Figure \ref{fig:tight-2-example} shows the non-zero edge values for a
submodular function $f$ (the values are multiplies of 1/3). It
\begin{figure}[htb]
\begin{center}\begin{tikzpicture}[scale=0.8]
\foreach \x in {1,...,8}{
  \draw[very thin](\x,0)--(\x,6.4);
}
\draw[->,very thin] (0,0)--(0,6.4);
\foreach \y in{1,...,6}{
  \draw[very thin](0,\y)--(8.4,\y);
}
\draw[->,very thin] (0,0)--(8.4,0);
\foreach\x in{0,...,8}{\draw(\x,-0.05) node[below]{\scriptsize\it\x};}
\foreach\y in{0,...,6}{\draw(-0.05,\y) node[left]{\scriptsize\it\y};}
\draw[line width=1.5pt](2,6.4)--(2,4)--(5,4)--(5,2)--(8.4,2);
\foreach \y in{4,5,6}{\draw[fill] (2,\y) circle(3pt);}
\foreach \x in{3,4,5}{\draw[fill] (\x,4) circle(3pt);}
\foreach \y in{2,3}{\draw[fill] (5,\y) circle(3pt);}
\foreach \x in{6,7,8}{\draw[fill](\x,2)circle(3pt);}
\begin{scope}[every node/.style={fill=white,font=\scriptsize,inner sep=1.5pt}]
\foreach \y in{0,1,2,3,4,5,6}{\draw (0.5,\y)node{5} (1.5,\y)node{5};}
\foreach\x in{0,1,2}{\foreach\y in{0.5,1.5,2.5,3.5}{
   \draw(\x,\y)node{5};}}
\foreach\x in{5,6,7,8}{
 \draw(\x,0.5)node{5} (\x,1.5)node{5} (\x,2.5)node{2} ;}
\foreach\x in{2.5,3.5,4.5}{\foreach\y in{0,1,2}{\draw(\x,\y)node{4};}
\draw(\x,3)node{3};}
\foreach\y in{4,5,6}{\draw(2.5,\y)node{2} (3.5,\y)node{2};}
\foreach\x in{3,4}{\draw(\x,0.5)node{5} (\x,1.5)node{5};}
\foreach\y in{2.5,3.5}{\draw(3,\y)node{4} (4,\y)node{3};}
\end{scope}
\end{tikzpicture}\end{center}

\kern -15pt
\caption{Submodular function by differences; values are multiplies of $1/3$}\label{fig:tight-2-example}
\end{figure}
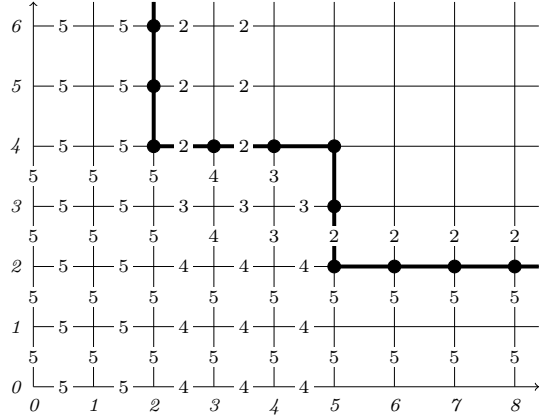%
realizes the bipartite access structure $\Gamma$ defined by the
points $(2,4)$ and $(5,2)$. Qualified and unqualified nodes are separated by
the solid line. The value of $f$ at any grid point is the sum of the differences
along any shortest ``Manhattan'' path from the point to the origin.
Conditions in (\ref{eq:shannon-diffs}) clearly hold.
For example, (\ref{eq:shannon-diffs}d) requires that values between adjacent
unqualified and qualified vertices should be at least one; such edges
are $(4,2)$--$(5,2)$, $(4,3)$--$(5,3)$, $(4,3)$--$(4,4)$, or
$(8,1)$--$(8,2)$. (\ref{eq:shannon-diffs}e) requires a difference of 1 or more for certain edge 
pairs such
as $(4,3)$--$(5,3)$--$(6,3)$, or $(1,y)$--$(2,y)$--$(3,y)$ for all $y\ge 4$.
There is only one square where (\ref{eq:shannon-diffs}f) applies, the one with diagonal points
$(4,4)$ and $(5,3)$. $\Gamma$ has a single step of width $\w=3$, thus Corollary
\ref{corr:simple} gives $\kappa(\Gamma)\ge 2-1/3$. As $H=V=5/3$ and $f$
realizes $\Gamma$, we also have $\kappa(\Gamma)\le 5/3$, thus
$\kappa(\Gamma)=5/3$. This construction generalizes for every single-step bipartite
access structure.

\begin{theorem}\label{thm:result1}
Suppose $\Gamma$ is defined by two points $(i_1,j_1)$ and $(i_2,j_2)$ where
$0<i_1$; $\w=i_2-i_1\ge \h=j_1-j_2$. If $\w\ge 2$ then
$\kappa(\Gamma)=2-1/\w$.
\end{theorem}
\begin{proof}
By Corollary \ref{corr:simple}, the Shannon complexity of $\Gamma$
is at least $2-1/\w$. 
The submodular function defined by the non-zero edge values on Figure 
\ref{fig:one-step}
\begin{figure}[htb]%
\def\+#1{${+}\!#1$}%
\begin{center}\begin{tikzpicture}[scale=0.8]
\foreach\x in {0,1,2,3,4,6,7,8,9,10}{
 \draw[very thin] (\x,-0.25)--(\x,7.25);}
\foreach\y in{0,1,2,4,5,6,7}{
 \draw[very thin](-0.25,\y)--(10.25,\y);}
\draw[dotted](-0.25,3)--(10.25,3) (5,-0.25)--(5,7.25);
\draw (4.7,7.3) node[above]{$\overbrace{\hbox to 4.6cm{}}^{\textstyle \w}$};
\draw(-0.3,4.2) node[left]{$\h\left\{\rule[-1.4cm]{0pt}{3.0cm}\right.$};
\draw[thick] (2,7.25)--(2,6) -- (8,6)--(8,2)--(10.25,2);
\foreach\x in{2,3,4,6,7,8}{\draw[fill](\x,6)circle(3pt); }
\draw[fill](2,7)circle(3pt) (9,2)circle(3pt) (10,2)circle(3pt);
\foreach\y in {2,4,5}{\draw[fill](8,\y) circle(3pt);}
\begin{scope}[every node/.style={fill=white,font=\scriptsize,inner sep=1.5pt}]
\foreach\x in{0.5,1.5}{\foreach\y in {0,1,2,4,5,6,7}{
  \draw(\x,\y)node{\+\w};
}}
\foreach\x in{2.5,3.5,4.5,5.5,6.5,7.5}{
  \draw(\x,5)node{\+1};
  \draw(\x,4)node{\+2};
  \draw(\x,2)node{\+\h};
  \draw(\x,1)node{\+\h};
  \draw(\x,0)node{\+\h};
}
\foreach\y in {6,7}{\foreach\x in{2.5,3.5,4.5,5.5,6.5}{
  \draw(\x,\y)node{\+0};}
}
\foreach\x in{8,9,10}{\foreach\y in{2.5,3.5,4.5}{
  \draw(\x,\y)node{\+0};
}
\foreach\y in{2.5,3.5,4.5,5.5}{
   \draw (7,\y)node{\+1}; \draw(3,\y)node{\+{\w\hbox{-}1}};
   \draw(6,\y)node{\+2};  \draw(4,\y)node{\+{\w\hbox{-}2}};
}}
\foreach\x in{0,1,2}{\foreach\y in{0,1,2,3,4,5}{
  \draw(\x,\y+0.5)node{\+\w};
}}
\foreach\x in{3,4,6,7,8,9,10}{
  \draw(\x,0.5)node{\+\w} (\x,1.5)node{\+\w};
}
\end{scope}
\draw[dotted](-0.25,3)--(10.25,3) (5,-0.25)--(5,7.25);
\end{tikzpicture}\end{center}

\kern -12pt
\caption{Single step access structure, values are multiplies of $1/\w$}\label{fig:one-step}
\end{figure}
has complexity $(2\w-1)/\w$, thus it gives the required upper bound. Edge
numbers are multiplies of $1/\w$ and numbers preceded by a $+$ sign should be
increased by $\w-1$, e.g., $+\w$ means $\w+\w-1=2\w-1$ (edge values at
the bottom left corner). Similarly, $+0=\w-1$, $+\h=\h+\w-1$, etc. The
bottom row and the leftmost column can be repeated until the bottom left
vertex becomes the origin. It is a routine to check that all conditions in
(\ref{eq:shannon-diffs}) actually hold.
\end{proof}

\begin{theorem}\label{thm:result2}
Let $\Gamma$ be a regular staircase with the same width and height
$\w=\h\ge 2$ such that for some $1\le k \le \ell_\Gamma$ the point
$(i_k,j_k)$ has positive coordinates. Then $\kappa(\Gamma)=2-1/\w$. 
\end{theorem}
\begin{proof}
The additional condition that $(i_k,j_k)$ is not on any of the coordinate 
axes guarantees that Corollary \ref{corr:simple} can be applied, and
gives $\kappa(\Gamma)\ge 2-1/\w$. For the other direction Figure
\ref{fig:regular-staircase} shows part of the non-zero edge values of a
\begin{figure}[htb]
\begin{center}\begin{tikzpicture}[scale=0.8]
\foreach\x in {0,...,10}{
 \draw[very thin] (\x,-0.9)--(\x,5.9);}
\foreach\y in{0,...,5}{
 \draw[very thin](-0.9,\y)--(10.9,\y);}
\draw (2.5,5.99) node[above]{$\overbrace{\hbox to 3.8cm{}}^{\textstyle \w}$};
\draw (7.5,5.99) node[above]{$\overbrace{\hbox to 3.8cm{}}^{\textstyle \w}$};
\draw(-0.99,2.5) node[left]{$\w\left\{\rule[-1.8cm]{0pt}{3.8cm}\right.$};
\draw[line width=1.5pt] (0,5.9)--(0,5) -- (5,5)--(5,0)--(10,0)--(10,-0.9);
\foreach\x in{0,...,5}{\draw[fill](\x,5)circle(3pt) (5+\x,0)circle(3pt);}
\foreach\y in{1,2,3,4}{\draw[fill](5,\y)circle(3pt); }
\begin{scope}[every node/.style={fill=white,font=\scriptsize,inner sep=1.8pt}]
\foreach\x in{1,...,9}{
  \foreach\y in{0.5,1.5,2.5,3.5}{
  \draw(9-\x,\y)node{\x};}}
\foreach\x in{5,6,7,8,9}{
  \draw(9-\x,4.5)node{\x} (14-\x,-0.5)node{\x};}
\foreach\x in{1,2,3,4}{\draw(4-\x,5.5)node{\x};}
\foreach\x in{0.5,...,3.5}{\draw(\x,5)node{4};}
\foreach\y in{5,6,7,8,9}{\foreach\x in{0.5,...,4.5}{
  \draw(\x,9-\y)node{\y};}}
\foreach\y in{1,2,3,4}{\foreach\x in{5.5,...,8.5}{
  \draw(\x,4-\y)node{\y};}}
\foreach\y in{0,...,5}{\draw(-0.5,\y)node{9};}
\foreach\x in{0,...,4}{\draw(\x,-0.5)node{9};}
\end{scope}
\end{tikzpicture}\end{center}

\kern-14pt
\caption{Regular staircase with same width and height $\w$}\label{fig:regular-staircase}
\end{figure}
submodular function for the regular staircase with $\h=\w=5$. Values are
multiplies of $1/\w$. The given pattern should be repeated by shifting it
down and right (up and left) by $\w$ until it fills the non-negative
quadrant. Conditions in (\ref{eq:shannon-diffs}) clearly hold. The pattern
easily generalizes for every regular staircase with equal width and height.
\end{proof}

\begin{theorem}\label{thm:matus}
Suppose all heights of the staircase $\Gamma$ are $1$, the first point
$(i_1,j_1)$ is not on the $y$-axis, and all widths are $\w_k\ge 2$. Then
\begin{equation}\label{eq:thm:matus}
  \kappa(\Gamma)\ge \kappa_0= 1+\frac{\ell_\Gamma-1}{
    \displaystyle 1+\sum_k \frac1{\w_k-1}}
\end{equation}
\end{theorem}
\begin{proof}
Let $\ell=\ell_\Gamma$ and denote the points $(i_1,j_1)$, \dots, $(i_\ell,j_\ell)$ by
$C_1$, \dots, $C_\ell$, see Figure \ref{fig:matus}.
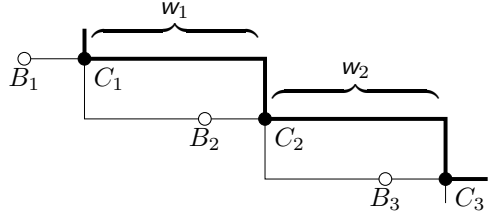
\begin{figure}[htb]
\begin{center}\begin{tikzpicture}[scale=0.8]
\draw[line width=1.5](1,0.5)--(1,0)--(4,0)--(4,-1)--(7,-1)--(7,-2)--(7.7,-2);
\draw[very thin](0,0)--(1,0)--(1,-1)--(4,-1)--(4,-2)--(7,-2)--(7,-2.4);
\draw[fill=white] (0,0) circle(3pt) (3,-1) circle(3pt)
       (6,-2) circle(3pt);
\draw[fill](1,0) circle(3pt) (4,-1) circle(3pt) (7,-2) circle(3pt);
\draw(2.5,0.15) node[above] {$\overbrace{\hbox to 2.2cm{}}^{\textstyle
\w_1}$};
\draw(5.5,-0.85) node[above] {$\overbrace{\hbox to
2.2cm{}}^{\textstyle\w_2}$};
\draw (0,0) node[below] {$B_1$} (1,0) node[below right]{$C_1$}
    (3,-1) node[below]{$B_2$} (4,-1) node[below right]{$C_2$} 
    (6,-2) node[below]{$B_3$} (7,-2) node[below right]{$C_3$};
\end{tikzpicture}\end{center}

\kern-14pt
\caption{Staircase with heights 1}\label{fig:matus}
\end{figure}
By Lemma \ref{lemma:increment} we have
\begin{equation*}
    H \ge C_1-B_1 \ge (C_2-B_2) + 1 - \frac {V-1}{\w_1-1},
\end{equation*} 
for each $1\le k\le \ell-1$
\begin{equation*}
    C_k-B_k \ge (C_{k+1}-B_{k+1})+1- \frac{V-1}{\w_k-1},
\end{equation*}
and finally
\begin{equation*}
    C_\ell-B_\ell \ge 1.
\end{equation*}
Adding them up we get
\begin{equation}\label{eq:sum-all}
    H \ge \ell-\sum_{k=1}^{\ell-1} \frac{V-1}{\w_k-1},
\end{equation}
or
$$
    H + V\sum \frac1{\w_k-1} \ge \ell+\sum \frac1{\w_k-1}.
$$
Consequently either $H$ or $V$ must be at least $\kappa_0$.
\end{proof}

Under an additional technical assumption the lower bound $\kappa_0$ in
Theorem \ref{thm:matus} is tight. The proof is by exhibiting an appropriate
discrete submodular function.

\begin{theorem}\label{thm:matus-sharp}
With the assumptions of Theorem \ref{thm:matus}, if, additionally,
$\w_k\ge\kappa_0$ for all widths, then $\kappa(\Gamma)=\kappa_0$.
\end{theorem}
\begin{proof}
The structure of non-zero edge values are sketched on Figure 
\ref{fig:construction}.
\begin{figure}[htb]
\def\+#1{${+}\!#1$}
\begin{center}\begin{tikzpicture}[scale=0.66]
\begin{scope}[very thin]
\foreach \y in {0.5,...,5.5}{\draw(-0.8,\y)--(2.8,\y);
\draw[dotted](2.8,\y)--(4.2,\y); \draw(4.2,\y)--(7.8,\y);
\draw[dotted](7.8,\y)--(9.2,\y); \draw(9.2,\y)--(12.8,\y);
\draw[dotted](12.8,\y)--(14.2,\y); \draw(14.2,\y)--(17.8,\y);
}
\foreach \x in {0,1,2,3,5,6,7,8,10,11,12,13,15,16,17,18}{\draw(\x-0.5,6.2)--(\x-0.5,-0.2); }
\end{scope}
\begin{scope}[line width=1.5pt]
\foreach \x in {0,1,2,3}{
  \draw(5*\x-1,5.5-\x)--(5*\x+1.5,5.5-\x)--(5*\x+1.5,4.5-\x)--(5*\x+3,4.5-\x);
}
\end{scope}
\foreach \x in {0,1,2,3}{
  \draw[fill] (5*\x-0.5,5.5-\x) circle(4pt)
   (5*\x+0.5,5.5-\x) circle(4pt)
   (5*\x+1.5,5.5-\x) circle(4pt)
   (5*\x+1.5,4.5-\x) circle(4pt)
   (5*\x+2.5,4.5-\x) circle(4pt);
}
\draw(4,6.3) node[above]{$\overbrace{\hbox to 3cm{}}^{\textstyle \w_{k-1}}$};
\draw(9,6.3) node[above]{$\overbrace{\hbox to 3cm{}}^{\textstyle \w_{k}}$};
\draw(14,6.3) node[above]{$\overbrace{\hbox to 3cm{}}^{\textstyle \w_{k+1}}$};
\begin{scope}[every node/.style={fill=white,font=\scriptsize,inner sep=1.5pt}]
\draw(-0.5,5)node{$\otimes$} (0.5,5) node{$1$};
 \foreach \x in {-1,0,1}{\draw(\x+0.5,4)node{$V$};}
 \foreach \x in {2,4}{\draw(\x+0.5,4)node{$\otimes$};}
 \draw(5.5,4)node{$1$};
 \foreach \x in {-1,0,1,2,4,5,6}{\draw(\x+0.5,3)node{$V$};}
 \foreach \x in {7,9}{\draw(\x+0.5,3)node{$\otimes$};}
 \draw(10.5,3)node{$1$};
 \foreach \x in {-1,0,1,2,4,5,6,7,9,10,11}{\draw(\x+0.5,2)node{$V$};}
 \foreach \x in {12,14}{\draw(\x+0.5,2)node{$\otimes$};}
 \draw(15.5,2)node{$1$};
 \foreach \x in {-1,0,1,2,4,5,6,9,10,11,12,14,15,16}{\draw(\x+0.5,1)node{$V$};}
 \draw(17.5,1)node{$\otimes$};
 \draw(1,5.5)node{$u$} 
   (0,4.5)node{\+u} (1,4.5)node{\+u};
 \foreach \y in {5.5,4.5}{
   \draw(2,\y)node{$u$} (5,\y)node{$u$} (6,\y)node{$z$};
 }
 \foreach \y in {5.5,4.5,3.5}{
   \draw(7,\y)node{$z$} (10,\y)node{$z$} (11,\y)node{$y$};
 }
 \foreach \y in {5.5,4.5,3.5,2.5}{
   \draw(12,\y)node{$y$} (15,\y)node{$y$} (16,\y)node{$x$}
     (17,\y)node{$x$};
 }
 \foreach \y in {3.5,2.5,1.5,0.5}{
    \draw(0,\y)node{\+u} (1,\y)node{\+u} (2,\y)node{\+z}
    (5,\y)node{\+z} (6,\y)node{\+z};
 }
 \foreach \y in {2.5,1.5,0.5}{
    \draw(7,\y)node{\+y} (10,\y)node{\+y} (11,\y)node{\+y};
 }
 \foreach \y in {1.5,0.5}{
    \draw(12,\y)node{\+x} (15,\y)node{\+x} (16,\y)node{\+x};
 }
 \draw(17,1.5)node{$x$};
\end{scope}
\end{tikzpicture}\end{center}

\vskip -12pt
\caption{Submodular function for the height 1 staircase}\label{fig:construction}
\end{figure}
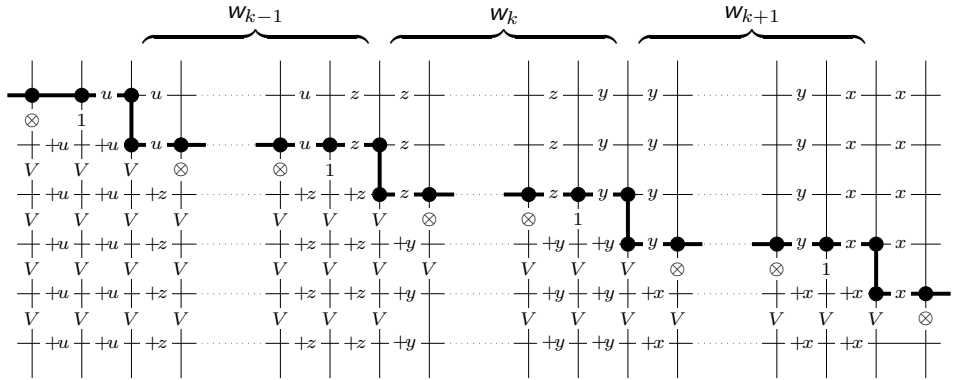
The $+$ symbol before $x$, $y$, etc., indicates $+1$, for example, ${+}y$
means $y+1$. The value $V$ is the ``vertical'' value between the origin and
$(0,1)$. There are sequences of vertical edges marked by $\otimes$ between a $V$ and a $1$
edge; their values should be computed so that they form an arithmetical
progression starting with $V$ and ending with $1$.

Assume $V\ge 1$ and that all edge values are non-negative. The consistency condition
in (\ref{eq:shannon-diffs}) clearly holds everywhere except around the edges marked by
$\otimes$. For the block under $\w_k$ the consistency requires
$$
    V+(\w_k-1)z = (\w_k-1)(y+1) + 1,
$$
that is,
\begin{equation}\label{eq:staircase-sharp}
    z= y + \frac{\w_k-V}{\w_k-1}.
\end{equation}
For the submodularity property we also need $z\ge y$, that is, $V\le 
\w_k$. If both of them are satisfied then all requirements in (\ref{eq:shannon-diffs}) hold.

After the last staircase step the horizontal edge values can be chosen to be
zero ($x=0$ in the figure). Other horizontal edge values are determined by
(\ref{eq:staircase-sharp}) and by the $+1$ increment, thus the edge between
$(0,0)$ and $(1,0)$ has the value
$$
H=1+\sum_{k<\ell}\frac{\w_k-V}{\w_k-1}=\ell-(V-1)\sum_{k<\ell}\frac{1}{\w_k-1}.
$$
Choosing $V=\kappa_0$ we get $H=\kappa_0$, which gives the required
submodular function.
\end{proof}

In Theorem \ref{thm:matus-sharp} the assumption that all steps have
width at least $\kappa_0$ is necessary. The next theorem shows 
that if some intermediate stepsize is smaller than $\kappa_0$, then the
Shannon complexity is strictly larger than $\kappa_0$. It happens, for
example, when the width sequence is $(3,3,2,3)$, when
$\kappa_0=2+7/49$ while the Shannon complexity is $2+7/34 > \kappa_0$.
\begin{theorem}\label{thm:small-stepsize}
With the assumptions of Theorem \ref{thm:matus} suppose that some intermediate
stepsize is smaller than $\kappa_0$. Then $\kappa(\Gamma)>\kappa_0$.
\end{theorem}
\begin{proof}
Use the notation of Figure \ref{fig:small-stepsize}.
Let $\Delta=\w_{k-1}+\w_k+\w_{k+1}-1$, this is the
distance between $A_2$ and $E_2$, or $Z_3$ and $D_3$.
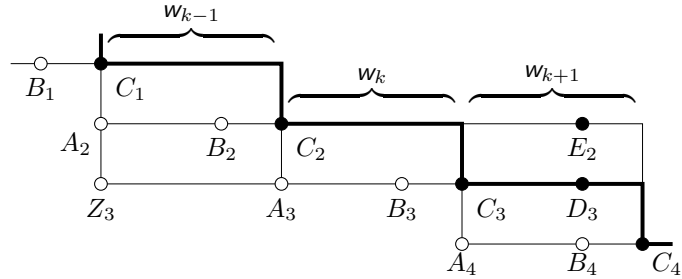
\begin{figure}[htb]
\begin{center}\begin{tikzpicture}[scale=0.8]
\draw[very thin](-3.5,1)--(-2,1)--(-2,0)--(1,0)--(1,-1)--
  (4,-1)--(4,-2)--(7,-2) (4,0)--(7,0)--(7,-1) (-2,0)--(-2,-1)--(1,-1);
\draw[line width=1.5pt](-2,1.5)--(-2,1)--(1,1)--(1,0)--(4,0)--
  (4,-1)--(7,-1)--(7,-2)--(7.5,-2);
\draw[fill=white] (-3,1) circle(3pt) (-2,0) circle(3pt) (0,0) circle(3pt)
      (1,-1) circle(3pt) (3,-1) circle(3pt) (-2,-1)circle(3pt)
      (4,-2) circle(3pt) (6,-2) circle(3pt);
\draw[fill](-2,1) circle(3pt) (1,0) circle(3pt) (4,-1) circle(3pt)
      (7,-2) circle(3pt) (6,0)circle(3pt) (6,-1)circle(3pt);
\draw(-0.5,1.15) node[above] {$\overbrace{\hbox to 2.2cm{}}^{\textstyle
\w_{k-1}}$};
\draw(2.5,0.15) node[above] {$\overbrace{\hbox to 2.2cm{}}^{\textstyle \w_k}$};
\draw(5.5,0.15) node[above] {$\overbrace{\hbox to
2.2cm{}}^{\textstyle\w_{k+1}}$};
\draw(-3,1) node[below=2pt]{$B_1$} (-2,1) node[below right=2pt]{$C_1$}
    (-2,0) node[below left]{$A_2$}
    (0,0) node[below=2pt] {$B_2$} (1,0) node[below right=2pt]{$C_2$}
    (6,0) node[below=2pt]{$E_2$} (-2,-1)node[below=2pt]{$Z_3$}
    (1,-1) node[below=2pt]{$A_3$} (3,-1) node[below=2pt]{$B_3$}
    (4,-1) node[below right=2pt]{$C_3$} (6,-1) node[below=2pt]{$D_3$}
    (4,-2) node[below]{$A_4$}
    (6,-2) node[below] {$B_4$} (7,-2) node[below right]{$C_4$};
\end{tikzpicture}\end{center}

\vskip -14pt
\caption{Case of a small stepsize}\label{fig:small-stepsize}
\end{figure}
The next two inequalities were actually proved in Lemma
\ref{lemma:increment}:
\begin{align*}
C_1-B_1 &\ge \frac{B_2-A_2}{\w_{k-1}-1}+1-\frac{V-1}{\w_{k-1}-1}\\
\frac{D_3-C_3}{\w_{k+1}-1} &\ge
   \frac{B_4-A_4}{\w_{k+1}-1} - \frac{V-1}{\w_{k+1}-1},
\end{align*}
and the following three ones follow from Lemma \ref{lemma:concave} easily:
\begin{align*}
  \frac{B_2-A_2}{\w_{k-1}-1} &\ge
    \frac{E_2-A_2}{\Delta}+1-\frac{\w_{k-1}}{\Delta}, \\
   \frac{D_3-Z_3}{\Delta} &\ge
    \frac{D_3-C_3}{\w_{k+1}-1}+\frac{\w_{k-1}+\w_k}{\Delta}, \\
   \frac{B_4-A_4}{\w_{k+1}-1} &\ge C_4-B_4.
\end{align*}
Finally, we have
$$
  \frac{E_2-A_2}{\Delta} \ge \frac{D_3-Z_3}{\Delta}-\frac{V}{\Delta}
$$
since $E_2\ge D_3$ and $A_2\le Z_3+V$.
Adding these inequalities up we get
\begin{align}\label{eq:improvement}
  C_1-B_1 &\ge (C_4-B_4) + 1-\frac{V-1}{\w_{k-1}-1} + {} \\
    & ~~ + 1-\frac{V-1}{w_{k+1}-1}
       + 1- \frac{V-1+(\w_{k-1}+\w_{k+1})}{\w_k-1+(\w_{k-1}+\w_{k+1})}
      .\nonumber
\end{align}
Looking back at the proof of Theorem \ref{thm:matus}, we see that using
(\ref{eq:improvement}), the right hand side of the inequality (\ref{eq:sum-all})
changes by
$$
    \frac{V-1}{\w_k-1}
-\frac{V-1+(\w_{k-1}+\w_{k+1})}{\w_k-1+(\w_{k-1}+\w_{k+1})}.
$$
When $\w_k<V$, this amount is positive (the second term is closer to 1
than the first one), thus the inequality in (\ref{eq:sum-all}) is strict.
Consequently $\kappa(\Gamma)>\kappa_0$ which proves the theorem.
\end{proof}

When some stepsize is below $\kappa_o$, then instead of (\ref{eq:sum-all})
one can use the improved estimate (\ref{eq:improvement})
to get a better lower bound on $\kappa(\Gamma)$. In some cases it gives the 
exact value, but not in every case.


\section{Some linear bipartite schemes}\label{sec:linear-examples}

We were able to create linear schemes with optimal complexity for a very
sparse set of non-ideal
bipartite access structures. For these structures the linear, the entropic,
and Shannon complexities are the same.

\begin{theorem}\label{thm:kings-pawns}
Let $\Gamma$ be the regular staircase with height $\h=1$, width $\w\ge 2$,
and length $\ell=\ell_\Gamma\ge 2$ 
such that the first point $(i_1,j_1)$ is not on the $y$-axis. There is a linear scheme for
$\Gamma$ with complexity
$$
   1+\frac{\ell-1}{\displaystyle 1+\frac{\ell-1}{\w-1}},
$$
which matches the lower bound $\kappa_0$ on $\kappa(\Gamma)$ from Theorem \ref{thm:matus}.
\end{theorem}
\begin{proof}
As explained in \cite[Theorem 5]{padro:optimization}, the scheme will be an
integer linear combination of schemes $h_1$ and $h_2$ defined below, both realizing
$\Gamma$. As both $h_1$ and $h_2$ can be represented over any finite field, this
combination gives the required linear secret sharing scheme.
The schemes distribute shares corresponding to some secret among the
participants in $N_1\cup N_2$ such that qualified subsets of the
regular staircase can recover the secret, while unqualified subsets
have no information on the secret. The share size (relative to the secret
size), however, will not be uniform. In the first scheme $h_1$ participants
in $N_1$ get single size shares, while participants in $N_2$ get shares of
size $\w$. In the second scheme $h_2$ it is the other way around:
participants in $N_2$ get single size shares, while participants in $N_1$
get shares of size $\ell$.

The idea is to combine several independent instances of these schemes --
assuming that they share the same secret size, which can be done in this
case. Executing $\alpha$ copies of $h_1$ and $\beta$ copies of $h_2$
distributes $\alpha{+}\beta$ many independent secrets. The total share size
of a participant from $N_1$ is $\alpha+\beta\cdot\w$ times the size of a
single secret, while for participants from $N_2$ this number is
$\alpha\cdot\ell+\beta$. Choosing $\alpha=\w-1$ and $\beta=\ell-1$ balances
these numbers to be
$$
   (\w-1)+(\ell-1)\w = (\w-1)\ell+(\w-1)=(\ell-1)(\w-1)+\ell+\w-2.
$$
Since the this combined scheme distributes $\alpha{+}\beta=\ell+\w-2$ many
secrets, its complexity is
$$
   1+\frac{(\ell-1)(\w-1)}{\ell+\w-2}=1+\frac{\ell-1}{\displaystyle 
        1+\frac{\ell-1}{\w-1}},
$$
matchning the value stated in the Theorem. It remains to describe the
two sub-schemes.

\textbf{Scheme 1} is an adaptation of the ideal bipartite scheme $\Gamma_3$ from
Section \ref{subsec:ideal-bipartite}, but see also \cite{padro:bipartite}. Let
$t=i_k+\w\cdot j_k$ (as $\Gamma$ is a regular staircase with height 1, this amount is
independent of $k$) and consider the following ideal bipartite access structure on
$N_1+\w\cdot N_2$ participants: a qualified set requires at least $i_1$
participants from the first set, at least $\w\cdot j_\ell$ participants from
the second set, and at least $t$ participants all together, see Figure
\ref{fig:ideal2}. To get $h_1$
form groups of size $\w$ from the second set, assigning all shares of a
group to a single participant from $N_2$.

In $h_1$ a the share size is one for a participant from $N_1$, and $\w$ for
a participant from $N_2$, as was claimed. To show that $h_1$ realizes the
access structure $\Gamma$, observe first that any qualified set in $h_1$ 
must have at least $i_1$ participants from the first set, and at least $j_\ell$
participants from the second set. When this holds, taking $i$ participants 
from $N_1$ and $j$ participants from $N_2$, $i\ge i_1$ and $j\ge j_\ell$
will hold, and this group forms a $h_1$-qualified set iff the
$i+j\w$ many shares they possess is above the threshold $t$. But this
happens iff $i\ge i_k$ and $j\ge j_k$ for some $k$, that is, if and only the 
point $(i,j)$ is in $\Gamma$.

\textbf{Scheme 2} is constructed as follows. The secret is a sum of two
independent values. The first one is distributed among members of $N_1$ using an $i_1$
out of $N_1$ threshold scheme. The second value is distributed using a $j_1$
out of $|N_2|+\ell-1$ threshold scheme. $|N_2|$ of the shares are given to
members of
$N_2$; the remaining $\ell-1$ shares are distributed among members
of $N_1$ as follows: one share is distributed using an $i_2$ out of $N_1$
threshold scheme, the second one by an $i_3$ out of $N_1$ threshold scheme,
and the last one by an $i_\ell$ out of $N_1$ threshold scheme. Every member
of $N_2$ gets a single share, while members of $N_1$ get $\ell$
shares.

Now we claim that $h_2$ also realizes $\Gamma$. A qualified set in $h_2$ must
recover both secret values. Recovering the first one requires at least $i_1$
members from $N_1$. Recovering the second value requires $j_1$ shares. Those
shares might come from $j_1$ members from $N_2$. They might also come from $j_1-1$
members from $N_2$, and the missing share can be recovered by $i_2$ members of
$N_1$. Similarly, the second
value can be recovered by $j_1-k$ members from $N_2$ and at least
$i_k$ members from $N_1$ for any $k\le\ell$. It shows that elements of
$\Gamma$ are qualified in $h_2$. The reverse follows from the fact that
recovering the second secret value
by $j_1-k$ members from $N_2$ requires at least $i_k$ members from $N_2$
\end{proof}

The rest of this section describes a linear scheme for a particular
bipartite access structure. For more clarity the construction uses vector spaces
over reals rather than over some finite field. This can be done as, by a
compactness argument, polymatroids representable over the reals are also
representable over 
some finite field whose characteristics can be chosen to be arbitrarily large.

In the constructions vectors contain unspecified variables. Their values
should be chosen to be \emph{generic}, by which we mean that considering all
vectors as rows of a huge matrix, if the determinant of any $k\times k$
submatrix is not a constant (that is, the determinant contains at least one
of the unspecified variable), then the determinant should differ from zero.
This can always be achieved, for example, by choosing all unspecified values
to be algebraically independent.

\begin{proposition}\label{const:30-11-03}
The complexity of the bipartite access
structure $\Gamma$ defined by the points $(0,3)$, $(1,1)$, $(3,0)$ is $3/2$.
\end{proposition}
\noindent
\begin{proof}
As $\kappa(\Gamma)\ge 3/2$ by Corollary \ref{corr:simple}, it is enough to
construct a linear scheme with this complexity.
We will work in the 7-dimensional vector space $\R^7$, as explained above. 
Participant $a\in
N_1$ will be assigned the 3-dimensional subspace $E_a$, and participant $b\in N_2$
will be assigned another 3-dimensional subspace $E_b$. The secret is a 2-dimensional subspace
$E_0$. This arrangement realizes the above bipartite access structure if
\begin{itemize}
\item[(a)] $E_0$ is in the linear hull of the subspaces assigned to any three participants
   from $N_1$ -- or any three participants from $N_2$.
\item[(b)] for every $a\in N_1$ and $b\in N_2$ the linear hull of 
       $E_a\cup E_b$ contains $E_0$;
\item[(c)] the linear hull of any two subspaces assigned to participants from $N_1$
     (or both form $N_2$) intersect $E_0$ trivially.
\end{itemize}
The subspace assigned to $a\in N_1$ and $b\in N_2$, respectively, will be spanned 
by the rows of the following matrices with seven columns:
$$
   E_a = \left(
   \begin{array}{c@{~}c@{~}c@{~~~}c@{~}c@{~}c@{~}c}
     \,1\,&\,0\,&\alpha_1& \alpha_3&\alpha_4 & 0 & 0 \\
     1&0&\alpha_2& 0 & 0 & \alpha_3& \alpha_ 4 \\
     0 & 1 & 0 & 0 & 0 & 0 & 0
   \end{array}\right)
   \qquad
   E_b=\left(
   \begin{array}{c@{~}c@{~}c@{~~~}c@{~}c@{~}c@{~}c}
     \,0\,&\,1\,&\beta_1& \beta_3&0&\beta_4  & 0 \\
     0&1&\beta_2& 0 & \beta_3&0& \beta_ 4 \\
     1 & 0 & 0 & 0 & 0 & 0 & 0
   \end{array}\right),
$$
while the secret space is spanned by the row vectors of
$$
   E_0=\left(
   \begin{array}{c@{~}c@{~}c@{~~~}c@{~}c@{~}c@{~}c}
     \zeta_1 & \zeta_2 & \zeta_3 & 0 & 0 & 0 & 0 \\
     \zeta_4 & \zeta_5 & \zeta_6 & 0 & 0 & 0 & 0
   \end{array}\right),
$$
where $\alpha_i$, $\beta_j$, $\zeta_k$ are the generic variables as explained
above.

(a) We show that the linear span of any three subspaces assigned to members of $N_1$
contain $E_0$; the case for $N_2$ is similar. The linear span of the row vectors
$$\left(
   \begin{array}{c@{~}c@{~}c@{~~~}c@{~}c@{~}c@{~}c}
     \,1\,&\,0\,&\alpha_1& \alpha_3&\alpha_4 & 0 & 0 \\[2pt]
     1&0&\alpha'_1& \alpha'_3&\alpha'_4 & 0 & 0 \\[2pt]
     1&0&\alpha''_1& \alpha''_3&\alpha''_4 & 0 & 0 
   \end{array}\right)
$$
contains the vector $(10\gamma_1\,0000)$ for some generic $\gamma_1$, and
similarly, the linear span of the row vectors
$$\left(
   \begin{array}{c@{~}c@{~}c@{~~~}c@{~}c@{~}c@{~}c}
     \,1\,&\,0\,&\alpha_2& 0&0&\alpha_3&\alpha_4 \\[2pt]
     1&0&\alpha'_2&0&0& \alpha'_3&\alpha'_4 \\[2pt]
     1&0&\alpha''_2&0&0& \alpha''_3&\alpha''_4 
   \end{array}\right)
$$
contains the vector $(10\gamma_2\,0000)$ for another generic $\gamma_2$.
(Actually, in both cases the same linear combination can be used.)
Thus the vectors $(100\,0000)$, $(010\,0000)$ and
$(001\,0000)$ are in the linear space spanned by $E_a$, $E_{a'}$ and
$E_{a''}$, and then so is $E_0$.

\smallskip
(b) One participant from the first group and one from the second one determine
the secret. The linear span of their subspaces contains the row vectors of
the matrix
$$\left(
   \begin{array}{c@{~}c@{~}c@{~~~}c@{~}c@{~}c@{~}c}
     \,0\,&0&\alpha_1& \alpha_3&\alpha_4 &0&0\\
     0&0&\alpha_2& 0 & 0 & \alpha_3& \alpha_ 4 \\
     0&0&\beta_1& \beta_3&0&\beta_4  & 0 \\
     0&0&\beta_2& 0 & \beta_3&0& \beta_ 4
   \end{array}\right)
$$
Taking their linear combination with coefficients
$(\beta_3,\beta_4,-\alpha_3,-\alpha_4)$ one gets the vector $(00\gamma\,0000)$
for some generic $\gamma$, thus $E_0$ is indeed inside their linear span.

\smallskip
(c) Let $E$ be the 3-dimensional subspace of vectors with the last four
coordinate equal to zero. As $E_0$ is a generic 2-dimensional subspace of $E$,
multiples of $(010\,0000)$ intersect $E_0$ trivially. We claim that the span 
of the remaining four vectors
$$\left(\begin{array}{c@{~}c@{~}c@{~~~}c@{~}c@{~}c@{~}c}
     \,1\,&\,0\,&\alpha_1& \alpha_3&\alpha_4 & 0 & 0 \\
     1&0&\alpha_2& 0 & 0 & \alpha_3& \alpha_4 \\
     1&0&\alpha'_1& \alpha'_3&\alpha'_4 & 0 & 0 \\[2pt]
     1&0&\alpha'_2& 0 & 0 & \alpha'_3& \alpha'_4
   \end{array}\right)
$$
intersect $E$ trivially. Indeed, no nontrivial linear combination makes the
last four coordinates zero as $(\alpha_3,\alpha_4)$ and $(\alpha'_3,\alpha'_4)$ are linearly
independent. Now $E_0$ is a subspace of $E$, therefore 
the linear span of $E_a$ and $E_b$ intersects $E_0$ trivially, as was
required.
\end{proof}

\section{Continuous submodular optimization}\label{sec:continuous}

Estimating the Shannon complexity of bipartite access structures can be
considered to be a discrete variant of a continuous submodular optimization
as has been discussed in, e.g., \cite{fujishige,hazan-kale}. The intuition
is scaling down the non-negative lattice so that the edge size becomes
negligible and take a bird's eye view. The continuous analog of a 
bipartite rank function is thus a real
function defined on the non-negative quadrant satisfying conditions
reflecting the conditions in (\ref{eq:shannon}) for discrete rank functions,
see Definition \ref{def:abcd}. These rank functions turn out to be
continuous and non-decreasing, consequently have both left and 
right partial derivatives, see Proposition \ref{prop:a2}.

An access structure $\Gamma$ specifies the qualified and unqualified points.
For the continuous case considered here $\Gamma$ is defined by a strictly
decreasing continuous curve connecting points on the coordinate axes.
Unqualified points are below the curve, and qualified points are above and
to the right of the curve.

For intuition how to specify whether a rank function $f\in\G$ \emph{realizes} an
access structure $\Gamma$ we turn to part d) of Lemma \ref{lemma:concave}.
It claims 
$$
   (B-A)/k \ge (D-C)/m +1
$$
assuming $A$, $B$, $C$, $D$ are, in this order, lattice points on a line,
where $A$ and $B$ are
unqualified and $C$, $D$ are qualified. Let $(u,v)$ be a boundary point
of $\Gamma$, and choose $A$, $B$, $C$, $D$ on a line parallel to the $x$ axis
so that $(u,v)$ is between $B$ and $C$. If all the points tend to $(u,v)$, the
fraction $(B-A)/k$ tends to the left partial derivative of $f$ at $(u,v)$
while $(D-C)/m$ tends to the right partial derivative. Thus, in the limit,
the above inequality says that $f^-_x(u,v)\ge f^+_x(u,v)+1$. Accordingly,
Definition \ref{def:realize} stipulates that the rank function $f\in\G$ realizes
$\Gamma$ if at every internal boundary point of $\Gamma$, both partial derivatives
of $f$ should drop by at least $1$.

Finally, the complexity of $f\in\G$, corresponding to the maximal share size
$\max\{H,V\}$ in the discrete case, is clearly should be $\max\{ f^+_x(0,0), f^+_y(0,0)\}$. The
continuous version of finding the Shannon complexity of a bipartite access
structure thus can be spelled out as follows.

\begin{optproblem}
For an access structure $\Gamma$, determined by the curve $\alpha$, 
determine the optimal complexity of continuous rank functions
realizing $\Gamma$.
\end{optproblem}

The rest of this section is organized as follows. First, the family of
continuous rank functions is defined, followed by propositions establishing
some of their basic properties. The main result is Theorem \ref{thm:cont0}
giving a general lower bound for this Optimization Problem in terms of the
curve $\alpha$. This bound is tight when $\alpha$ is linear. The section
concludes with a few remarks and open problems.

\begin{definition}\label{def:abcd}
The family $\G$ of \emph{continuous bipartite rank functions} consists of
real functions $f$ defined on the non-negative
quadrant $[0,+\infty)^2$ satisfying conditions a)--d) below.
\begin{itemize}
\item[a)] $f$ is pointed, that is, $f(0,0)=0$;
\item[b)] $f$ is non-decreasing: for $0\le u_1\le u_2$  and $0\le v_1\le v_2$ we
have $f(u_1,v_1)\le f(u_2,v_2)$;
\item[c)] $f$ is concave separately in both coordinates: for $0< t < 1$ and
$\hat t=1-t$,
\begin{align*}
   f(tu_1+\hat t u_2,v) &\ge tf(u_1,v)+\hat t f(u_2,v), \\
   f(u,tv_1+\hat t v_2) &\ge tf(u,v_1)+\hat t f(u,v_2);
\end{align*}
\item[d)] $f$ is submodular: for $~0\le u_1\le u_2$ and $ 0\le v_1\le v_2$ 
\begin{equation}\label{eq:c4}
  f(u_1,v_2)+f(u_2,v_1)\ge f(u_1,v_1)+f(u_2,v_2).
\end{equation}
\end{itemize}
\end{definition}

The class $\G$ is closed for non-negative linear combinations.
Moreover, if $f\in\G$ and $M$ is a non-negative constant, then
$\min(f,M)\in\G$. Consequently $f(u,v)=\min\{c_u\,u+c_v\,v, M\}$ is in $\G$
for every positive $c_u$, $c_v$ and $M$.

The right and left partial derivatives of $f\in G$, if exist, are denoted by
$f^+_x$, $f^+_y$ and $f^-_x$, $f^-_y$, respectively. Some properties of
functions in $\G$, similar to those of discrete bipartite rank functions,
follow from the definition above.

\begin{proposition}\label{prop:a1}
$f$ is concave and increasing along any positive direction:
if $0\le u_1\le u_2$, $0\le v_1\le v_2$, $0\le t\le 1$ and $\hat t=1-t$, then
$$
  f(tu_1+\hat tu_2,tv_1+\hat tv_2)\ge tf(u_1,v_1)+\hat t f(u_2,v_2).
$$
\end{proposition}
\begin{proof}
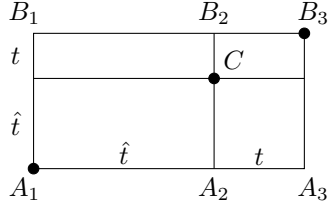
\begin{figure}[tb]
\hfill\begin{tikzpicture}[x={(1.2,0)},y={(0,0.6)}]
\draw(0,0)--(3,0)--(3,3)--(0,3)--cycle;
\draw (2,0)--(2,3) (0,2)--(3,2);
\draw[fill] (0,0) circle (2pt) (3,3) circle (2pt) (2,2) circle (2pt);
\draw (-0.1,0) node[below] {$A_1$} (2,0) node[below] {$A_2$} (3.1,0) node[below]
{$A_3$};
\draw (-0.1,3) node[above] {$B_1$} (2,3) node[above] {$B_2$} (3.1,3) node[above]
{$B_3$};
\draw (2,2) node[above right] {$C$};
\draw (1,0.35) node {$\hat t$} (2.5,0.25) node {$t$};
\draw (-0.2,1) node {$\hat t$} (-0.2,2.5) node {$t$};
\end{tikzpicture}\hfill\hbox{}\par
\caption{$f\in\G$ is concave in positive directions}\label{fig:pos-concave}
\end{figure}
Refer to Figure \ref{fig:pos-concave} where $A_1$, $B_2$ are the points
with coordinates $(u_1,v_1)$ and $(u_2,v_2)$, respectively. Concavity along
the $x$ and $y$ coordinates give
\begin{align*}
   t f(A_1)+\hat tf(A_3) & \le f(A_2), \\
   t f(B_1)+\hat tf(B_3) &\le  f(B_2), \\
   t f(A_2)+\hat tf(B_2) & \le f(C).
\end{align*}
Multiplying the first inequality by $t$, the second one by $\hat t$, and using
(\ref{eq:c4}) to get
$$
    f(A_1)+f(B_3)\le f(B_1)+f(A_3),
$$
the required inequality follows.
\end{proof}

Note that the function $f(u,v)=\min(u,1)\cdot\min(v,1)$ satisfies properties a)--c)
of Definition \ref{def:abcd}, while does not satisfy d) as $f$ is not concave in the 
$(1,1)$ direction.

\begin{proposition}\label{prop:a2}
Partial derivatives of $f\in\G$ exist (allowing the value $+\infty$ at the
boundary), they are non-negative and non-increasing in both coordinates.
\end{proposition}
\begin{proof}
For example, if $u_1<u_2$, then
\begin{align*}
   f^+_y(u_1,v)&=\lim_{v'\to v+0} \frac{f(u_1,v')-f(u_1,v)}{v'-v} \\
          &\ge\lim_{v'\to v+0} \frac{f(u_2,v')-f(u_2,v)}{v'-v}=f^+_y(u_2,v),
\end{align*}
where the inequality holds by (\ref{eq:c4}).
\end{proof}

Similar reasoning gives
\begin{proposition}\label{prop:a3}
If $u\to u'-0$, then $f^+_x(u,v) \to f^-_x(u',v)$;
if $v\to v'+0$, then $f^+_y(u,v) \to f^+_y(u,v')$; and
similarly for other cases. \qed
\end{proposition}

\begin{definition}\label{def:realize}
The rank function $f\in\G$ \emph{realizes} the access structure $\Gamma$ 
if at every internal boundary point
$(u,v)$ of $\Gamma$ (that is, when both $u$ and $v$ are positive) we have
\begin{equation}\label{eq:realize}
\begin{array}{r@{\;\ge\;}l}
  f^-_x(u,v) & 1+f^+_x(u,v), \\[3pt]
  f^-_y(u,v) & 1+f^+_y(u,v).
\end{array}\end{equation}
\end{definition}

We consider only access structures which are defined by the graph
of a continuous, strictly decreasing curve $\alpha$ such that
$\alpha(0)=a>0$ and $\alpha(b)=0$ for some $b>0$. The point $(x,y)$ is
qualified if either $x\ge b$, or if $x\ge 0$ and $y\ge \alpha(x)$. In this
case internal points of the boundary are the points $(x,\alpha(x))$ for
$0<x<b$.

\begin{lemma}\label{lemma:cont1}
Suppose $\alpha$ is as above, it is derivable everywhere and $f\in\G$ satisfies the 
constraints {\upshape(\ref{eq:realize})} in internal points of the graph
of $\alpha$. Then
$f^+_x(0,0)\ge \sup\,\{-\alpha'(v):0<v<b\}$, where $\alpha'$ is the
derivative of $\alpha$.
\end{lemma}
\begin{proof}
Let $0\le u<v$, in this case $\alpha(v)<\alpha(u)$. Then
$$
    \left[v-u\right] f^+_x(u,\alpha(v)) \ge
    f(v,\alpha(v))-f(u,\alpha(v))
$$
We also have
\begin{align*}
   f(v,\alpha(u))-f(u,\alpha(u)) &\ge 0, \\[2pt]
   f(u,\alpha(u))-f(u,\alpha(v)) &\ge \left[\alpha(u)-\alpha(v)\right]
         f^-_y(u,\alpha(u)) \ge {} \\
          &\ge
\left[\alpha(u)-\alpha(v)\right]\left(1+f^+_y(u,\alpha(u))\right),
\end{align*}
and
$$
   \left[\alpha(u)-\alpha(v)\right] f^+_y(v,\alpha(v)) \ge
   f(v,\alpha(u))-f(v,\alpha(v)).
$$
Adding them up we get
$$
   \left[v-u\right]f^+_x(u,\alpha(v))
    \ge
   \left[\alpha(u)-\alpha(v)\right]
   \left(1+f^+_y(u,\alpha(u))-f^+_y(v,\alpha(v))\right).
$$
Now $f^+_y$ is non-increasing in both directions, thus 
$f^+_y(u,\alpha(u))\ge f^+_y(v,\alpha(u))$, which means
$$
  f^+_x(u,\alpha(v)) \ge \frac{\alpha(u)-\alpha(v)}{v-u}
    \left[1+f^+_y(v,\alpha(u))-f^+_y(v,\alpha(v))\right].
$$
Limiting $u\to v-0$ the left hand side becomes $f^-_x(v,\alpha(v))$,
and on the right hand side we have $f^+_y(v,\alpha(u))\to
f^+_y(v,\alpha(v))$. Therefore $f^-_x(v,\alpha(v))\ge -\alpha'(v)$, which
immediately gives the claim.
\end{proof}

Recall that the complexity of the rank function $f\in\G$ is $\max\{ 
f^+_x(0,0), f^+_y(0,0)\}$.

\begin{theorem}\label{thm:cont0}
Let $\alpha$ is strictly decreasing, derivable everywhere,
$\alpha(0)=a>0$ and $\alpha(b)=0$ for some $b>0$. 
The complexity of every $f\in\G$ realizing the access structure defined
by $\alpha$ is at least
$$
   \sup\,\{ -\alpha'(v),\; -1/\alpha'(v) :~ 0<v<b \}.
$$
\end{theorem}

\begin{proof}
The inequality $f^+_y(0,0) \ge \sup\{-1/\alpha'(v): 0<v<b\}$
is equivalent to 
$$
  f^+_x(0,0)\ge \sup\{-\alpha'(v):0<v<b\},
$$
proved in Lemma \ref{lemma:cont1}, by symmetry that exchanges the arguments of 
$f$ and $\alpha$ with its inverse $\alpha^{-1}$. Consequently the maximum of
the two $\sup$s is a lower bound on the complexity.
\end{proof}

The bound provided by Theorem \ref{thm:cont0} is tight when $\alpha$ is
linear, it is attained by $f(u,v)=\min\{c_u\,u+c_v\,v,M\}$ for some positive
$c_u$, $c_v$, $M$. It is interesting to note that the proof of Lemma
\ref{lemma:cont1} used only the local behavior of $f$ at the curve points
$(u,\alpha(u))$ without considering any global accumulation effect. It would
be interesting to know whether this is typical or not. There are many other
open questions, like: when the above bound is tight; whether the infimum is
attained, or uniquely attained; whether the solution depends continuously on
$\alpha$; what is the relation between the discrete and continuous cases;
etc.

\section*{Acknowledgment}
The work of the first author was partially supported by the ERC Advanced Grant ERMiD. The work of third author was supported by 
the Spanish Government
under Project PID2019-109379RB-I00.


\end{document}